\documentclass[11p, reqno]{amsart}
\usepackage{amssymb,amsthm,amsmath,mathrsfs,bm,braket,marginnote}
\numberwithin{equation}{section}
\pdfoutput=1
\usepackage{appendix} 
\usepackage[utf8]{inputenc}
\usepackage[T1]{fontenc}
\usepackage{microtype}

\usepackage{fourier}
\usepackage{caption}
\usepackage[colorlinks=true, pdfstartview=FitV, linkcolor=blue, citecolor=blue, urlcolor=blue]{hyperref}
\usepackage[numbers,sort&compress]{natbib}
\usepackage{enumitem}

\newcommand{\jump}[1]{\bigl[ #1 \bigr]}

\theoremstyle{plain}
\newtheorem{proposition}{Proposition}[section]
\newtheorem{corollary}[proposition]{Corollary}

\newtheorem{lemma}[proposition]{Lemma}
\newtheorem{theorem}[proposition]{Theorem}

\theoremstyle{definition}
\newtheorem{definition}[proposition]{Definition}
\newtheorem{example}[proposition]{Example}
\newtheorem{remark}[proposition]{Remark}

\usepackage[usenames,dvipsnames]{xcolor}
\usepackage{pgf}
\usepackage{pgfplots}
\usepackage{tikz}
\usetikzlibrary{arrows,calc}
\usepackage{verbatim}
\usetikzlibrary{decorations.pathreplacing,decorations.pathmorphing}
\usepackage{multicol}  
\usepackage{multirow}
\usetikzlibrary{%
    decorations.pathreplacing,%
    decorations.pathmorphing%
}

\newcommand{\R}{\mathbf{R}}

\newcommand{\D}{\mathcal{D}}

 \newcommand{\avg}[1]{\bigl\langle #1 \bigr\rangle}
 \newcommand{\sgn}[1]{\operatorname{sgn} #1}

\newcommand{\abs}[1]{\left\lvert #1 \right\rvert}
\newcommand{\hoc}[1]{#1^+}
\newcommand{\floor}[1]{\left\lfloor #1 \right\rfloor}

\usepackage{tikz}
\tikzset{
  source/.style={circle,draw=black!100,fill=black!50,inner sep = 0,minimum size=2mm},
  sink/.style={circle,draw=black!100,fill=white,inner sep = 0,minimum size=2mm}
}

\usepackage{microtype}
\begin{document}

%\title[Peakon Problem for the modified Camassa-Holm equation] {The solution to the peakon problem for the modified Camassa-Holm equation and multi-point Pad\'{e} approximants}
\title[Vibrations of an elastic bar and mCH equations]{Vibrations of an elastic bar, isospectral deformations, and modified Camassa-Holm equations}

\author{Xiang-Ke Chang}
\address{ LSEC, ICMSEC, Academy of Mathematics and Systems Science, Chinese Academy of Sciences, P.O.Box 2719, Beijing 100190, PR China; and School of Mathematical Sciences, University of Chinese Academy of Sciences, Beijing 100049, PR China.}
\email{changxk@lsec.cc.ac.cn}
%\thanks{}

\author{Jacek Szmigielski}
\address{Department of Mathematics and Statistics, University of Saskatchewan, 106 Wiggins Road, Saskatoon, Saskatchewan, S7N 5E6, Canada.}
\email{szmigiel@math.usask.ca}

%\date{June 16, 2012}
%\date{\today}

\begin{abstract}
    Vibrations of an elastic rod are described by a Sturm-Liouville system.
  We present a general discussion of isospectral (spectrum preserving)
  deformations of such a system.  We interpret one family of such deformations in terms of a two-component modified Camassa-Holm equation (2-mCH) and solve completely its dynamics for
  the case of discrete measures (multipeakons).  We show that the
  underlying system is Hamiltonian and prove its Liouville integrability.
  The present paper generalizes our previous work on interlacing multipeakons of the 2-mCH and multipeakons of the 1-mCH.
We give a unified approach to both equations, emphasizing certain natural family of interpolation problems germane to  the solution of the inverse problem for 2-mCH as well as to this type of  a Sturm-Liouville system with singular coefficients.
\end{abstract}
% Weak solutions, peakons, Weyl function, inverse problem, continued fractions, Pad{\'e} approximation.

%35D30, %Weak solutions
%35Q51, % Solitons
%34K29, % Inverse problems
%37J35, % Completely integrable systems, topological structure of phase space, integration methods
%35Q53, % KdV-like equations
%34B05, % Linear boundary value problems
%41A21. % Padé approximation

% 34  ODE
% 34B Boundary value problems
% 34K Functional-differential and differential-difference equations

% 35  PDE
% 35Q Equations of mathematical physics and other areas of application
% 35D30 Weak solutions 
% 37  Dynamical systems and ergodic theory
% 37J Finite-dimensional Hamiltonian, Lagrangian, contact, and nonholonomic systems

% 41  Approximations and expansions

%\subjclass[2010]{Primary 34A55} % 34A55 Ordinary differential equations - General theory - Inverse problems

\maketitle

\tableofcontents

\section{Introduction}
One of the most important applications of the Sturm-Liouville
systems is provided by the longitudinal vibrations of
an elastic bar of stiffness $p$ and density $\rho$ \cite[Chapter 10.3]{birkhoff-rota}.  The longitudinal
displacement $v$ satisfies the wave equation
\begin{equation}
\rho(x) \frac{\partial^2 v}{\partial t^2}=\frac{\partial}{\partial x}[p(x) \frac{\partial v}{\partial x}],
\end{equation}
which after the separation of variables $v=u(x)\cos\omega t$  leads to
\begin{equation}\label{eq:SL}
D_x [p(x)D_x u]+\omega^2 \rho(x) u=0,
\end{equation}
where $D_x=\frac{d}{dx}$.
In a different area of applications, in geophysics,
the \textit{Love waves}  were proposed by  an early
20 century British geophysicist Augustus Edward Hugh Love who predicted
the existence of \textit{horizontal surface waves} causing Earth shifting
during an earthquake.  The wave amplitudes of these
waves satisfy \cite{aki-richards}
\begin{equation} \label{eq:Love}
D_x (\mu D_x) u+(\omega^2 \rho -k^2 \mu)u=0,  \quad 0<x<\infty,
\end{equation}
where $\mu$ is the \textit{sheer modulus}, $x$ is the depth below the
Earth surface and the boundary
conditions are $D_x u(0)=u(\infty)=0$ which can be interpreted as
the Neumann condition on one end and the Dirichlet condition on the other.
In applications to geophysics the frequency $\omega$ is fixed and
the phase velocity is $\omega/k$. In particular, in
the infinity speed limit ($k=0$), we obtain the same Sturm-Liouville system
as in \eqref{eq:SL}.  In either case the
problem can conveniently be written as
 the first order system:
\begin{equation}\label{eq:Lovesys}
D_x \Phi=\begin{bmatrix}0& n\\-\omega^2 \rho&0 \end{bmatrix}
\Phi,
\end{equation}
where $n=\frac1\mu$ , $\Phi=\begin{bmatrix} \phi_1\\ \phi_2 \end{bmatrix}$, $\phi_1=\phi, \, \phi_2=\mu D_x \phi$.
In the present paper we study \eqref{eq:Lovesys}
 on the whole
real axis and impose the boundary conditions $\phi_1(-\infty)=\phi_2(\infty)=0$
which can be interpreted as the Dirichlet condition at  $-\infty$ and the Neumann condition at $+\infty$.
Our motivation to study this problem comes
from yet another area of applied mathematics dealing
with integrable nonlinear partial differential equations.   To focus our discussion we begin by considering
the nonlinear partial differential equation
\begin{equation}\label{eq:m1CH}
m_t+\left((u^2-u_x^2) m\right)_x=0,  \qquad
m=u-u_{xx},
\end{equation}
which is one of many variants of the famous Camassa-Holm equation (CH) \cite{camassa1993integrable}:
\begin{equation} \label{eq:CH}
m_t+u m_x +2u_x m=0, \qquad  m=u-u_{xx},
\end{equation}
for the shallow water waves.  We will call \eqref{eq:m1CH} the mCH equation for short.
The history of the mCH equation is long and convoluted: \eqref{eq:m1CH} appeared in the papers of Fokas \cite{fokas1995korteweg}, Fuchssteiner \cite{fuchssteiner1996some},  Olver and Rosenau\cite{olver1996tri} and was, later, rediscovered by Qiao
 \cite{qiao2006new,qiao2007new}.

Subsequently, in \cite{song2011new},  Song, Qu and Qiao proposed a natural two-component generalization of \eqref{eq:m1CH}
\begin{equation}\label{eq:m2CH}
\begin{aligned}
  m_t&+[(u-u_x)(v+v_x)m]_x=0,\\
  n_t&+[(u-u_x)(v+v_x)n]_x=0,\\
     &m=u-u_{xx},\qquad n=v-v_{xx},
\end{aligned}
\end{equation}
which, for simplicity, we shall call the 2-mCH. Formally, the 2-mCH reduces to the mCH when $v=u$.

We are interested in the class of non-smooth solutions of \eqref{eq:m2CH} given by the \textit{peakon ansatz} \cite{camassa1993integrable,chang2016multipeakons,lundmark-szmigielski:GX-inverse-problem}:
\begin{equation} \label{eq:peakonansatz}
u=\sum_{j=1}^N m_j (t)e^{-\abs{x-x_j(t)}}, \qquad  v=\sum_{j=1}^N n_j (t)e^{-\abs{x-x_j(t)}},
\end{equation}
where all smooth coefficients $m_j(t),n_j(t)$ are taken to be positive,
and hence
\begin{equation*}
~m=u-u_{xx}=2\sum_{j=1}^N m_j \delta_{x_j},\qquad n=v-v_{xx}=2\sum_{j=1}^N n_j \delta_{x_j}
\end{equation*}
are positive discrete measures.

For the above ansatz, \eqref{eq:m2CH} can be viewed as a distribution equation, requiring in particular that we
define the products $Q m$ and $Qn$, where
\begin{equation} \label{eq:Q}
Q=(u-u_x)(v+v_x).
\end{equation}
It is shown in Appendix \ref{lax_m2ch} that the
choice consistent with the Lax integrability discussed in Section \ref{sec:Lax} is to take $Qm$, $Qn$ to mean $\langle Q \rangle m$, $\langle Q \rangle n$ respectively,
where $\langle f \rangle$ denotes the average function (the arithmetic average
of the right and left limits). Substituting \eqref{eq:peakonansatz} into \eqref{eq:m2CH} and using the
multiplication rule mentioned above
leads to the system of ODEs:
\begin{subequations}\label{eq:m2CH_ode}
\begin{align}
&\dot m_j=0, \qquad \dot n_j=0,\\
&\dot x_j=\avg{Q}(x_j).
\end{align}
\end{subequations}

In the present paper, we shall develop an inverse spectral approach to solve the peakon ODEs
\eqref{eq:m2CH_ode} and hence \eqref{eq:m2CH} under the following assumptions:
\begin{enumerate}
\item all $m_k,n_k$ are positive,
\item the initial positions are assumed to be ordered as $x_1(0)<x_2(0)<\cdots<x_n(0)$.
\end{enumerate}
We emphasize that the second condition is not restrictive since it
can be realized by relabeling positions as long as positions $x_j(0)$ are distinct.

The present paper generalizes our previous work
on interlacing multipeakons of the 2-mCH in \cite{chang2016multipeakons} and multipeakons of the 1-mCH in \cite{chang-szmigielski-m1CHlong}.
It is worth mentioning, however, that the technique of the present paper
is a modification of the one employed in \cite{chang-szmigielski-m1CHlong} and is distinct from the inhomogeneous string approach
adapted in \cite{chang2016multipeakons}.  As a result
we give a unified approach to both equations; this is accomplished by  putting  common interpolation problems front and center
of the solution to the inverse problem for \eqref{eq:Lovesys}.
Moreover, by solving \eqref{eq:m2CH_ode}, we furnish a family of isospectral flows
for the Sturm-Liouville system \eqref{eq:Lovesys}.
The full explanation of the connection between \eqref{eq:Lovesys} and
\eqref{eq:m2CH} is
reviewed in the following sections.

\section{The Lax formalism: the boundary value problem} \label{sec:Lax}
The Lax
pair for \eqref{eq:m2CH} can be written:
\begin{equation}
  \frac{\partial}{\partial x}
  \Psi_x=\frac12U \Psi, \quad  \Psi _t =\frac12 V \Psi, \quad  \Psi=\begin{bmatrix} \Psi_1\\\Psi_2 \end{bmatrix} ,
  \label{2chlax}
\end{equation}
where
\begin{equation}
\label{2chlax_uv}
  U =
  \begin{pmatrix}
    -1 & \lambda m \\
    -\lambda n & 1 \\
  \end{pmatrix}
  ,\quad
  V=
  \begin{pmatrix}
    4\lambda^{-2}+Q & -2\lambda^{-1}(u-u_x)-\lambda mQ\\
    2\lambda^{-1}(v+v_x)+\lambda nQ & -Q
  \end{pmatrix}
  ,
\end{equation}
with $Q=(u-u_x)(v+v_x)$.  This form of the Lax pair is a slight modification (in particular, $V$ have slightly different diagonal terms) of the original Lax pair in \cite{song2011new}.  The modification
is needed for consistency with the boundary value problem to be discussed below in Remark \ref{rem:V}.

We recall that for smooth
solutions the role of the Lax pair is
to provide the Zero Curvature representation
$\frac{\partial U}{\partial t}-\frac{\partial V}{\partial x}+\tfrac12 [U,V]=0$ of
the original non-linear partial differential equation, in our case
\eqref{eq:m2CH}.  In the non-smooth case the situation is more subtle as explained in Appendix \ref{lax_m2ch}.

Following \cite{chang2016multipeakons} we perform
the gauge transformation $\Phi=\textrm{diag}(\frac{e^{\frac x2}}{\lambda}, e^{-\frac x2}) \Psi$ which leads to a simpler $x$-equation
\begin{equation}\label{eq:xLax}
\Phi_x=\begin{bmatrix}0 & h\\
-z g& 0 \end{bmatrix} \Phi, \qquad    g=\sum_{j=1} ^Ng_j \delta_{x_j}, \qquad h=\sum_{j=1} ^Nh_j \delta_{x_j},
\end{equation}
where $g_j=n_j e^{-x_j}, \, h_j =m_j e^{x_j}, \, z=\lambda^2 $, and thus  $g_jh_j=m_jn_j$.
We note that \eqref{eq:xLax}, that is the $x$ member of the Lax pair, has the form of  the Sturm-Liouville
problem given by \eqref{eq:Lovesys}, provided we specify the boundary conditions.  Our initial goal is to solve \eqref{eq:xLax}
subject to boundary conditions $\Phi_1(-\infty)=0, \, \Phi_2(+\infty)=0$ which are chosen in such a way as to remain invariant under the flow in
 $t$ whose infinitesimal change is generated by the matrix $V$ in the
 Lax equation.

Since the coefficients in the boundary value problem are distributions (measures), to make
\begin{equation}\label{eq:xLaxBVP}
\Phi_x=\begin{bmatrix}0 &h\\
-z g& 0 \end{bmatrix} \Phi, \qquad \Phi_1(-\infty)=\Phi_2(+\infty)=0,
\end{equation}
well posed, we need to define the multiplication of the measures $h$ and $g$ by $\Phi$.  As suggested by the results in Appendix \ref{lax_m2ch}
we require that $\Phi$ be left continuous and we subsequently define the terms $\Phi_{a}\delta_{x_j}=\Phi_a(x_j)\delta_{x_j}, a=1,2$.  This choice makes the Lax pair well defined as a distributional Lax pair and, as it is shown in the Appendix \ref{lax_m2ch},
the compatibility condition of the $x$ and $t$ members of the Lax pair indeed implies \eqref{eq:m2CH_ode}.  The latter result is more subtle than
a routine check of compatibility for smooth Lax pairs.

Since the right-hand side of \eqref{eq:xLaxBVP} is
zero on the complement of the support of $g$ and $h$, which in our case
consists of points $\{x_1, \dots, x_N\}$, the solution $\Phi$ is a piecewise constant function, which solves a finite difference equation.
\begin{lemma}\label{lem:forwardR}
Let $q_k=\Phi_1(x_{k}+), \, ~p_k=~\Phi_2(x_{k}+),$ then the difference form of the
boundary value problem \eqref{eq:xLaxBVP} reads:

\begin{equation}\label{dstring}
\begin{bmatrix}
  q_{k}\\
  p_{k}
\end{bmatrix}
=T_k\begin{bmatrix}
  q_{k-1}\\
  p_{k-1}
\end{bmatrix}, \qquad
T_k=\begin{bmatrix}
  1& h_k\\
  -z g_k&1
\end{bmatrix}, \qquad 1\leq k\leq N,
\end{equation}
where $q_0=0, \, p_0=1$, and the boundary condition on the right end (see \eqref{eq:xLaxBVP}) is satisfied whenever $ p_{N}(z)=0.  $

\end{lemma}
By inspection we obtain the following corollary.
\begin{corollary} \label{cor:pq-degrees}

$q_{k}(z)$  is a polynomial of degree $\lfloor\frac{k-1}{2}\rfloor$ in $z$, and $p_{k}(z)$ is a polynomial of degree $\lfloor\frac{k}{2}\rfloor$, respectively.

\end{corollary}

\begin{remark} \label{rem:V}
Note
that $\det(T_k)=1+zh_kg_k=1+zm_kn_k\neq 1$.  In other words the setup we are developing goes beyond an $SL_2$ theory.   In order to understand
the origin of this difference we go back to the original Lax pair \eqref{2chlax_uv}.  If we assumed the matrix $V$ to be traceless,
with some coefficient $\alpha(\lambda)$ on the diagonal,
then in the asymptotic region $x>>0$ the second
equation in \eqref{2chlax} would
read
\begin{equation*}
\begin{bmatrix} \dot a e^{-x/2}\\ \dot b e^{x/2} \end{bmatrix}
=\tfrac12 \begin{bmatrix} \alpha(\lambda)& -4\lambda^{-1} u_+ e^{-x}\\
0 & -\alpha(\lambda) \end{bmatrix} \begin{bmatrix} ae^{-x/2}\\ be^{x/2}\end{bmatrix} ,
\end{equation*}
where $u(x)=u_+ e^{-x}$ in the asymptotic region.
The simplest way to implement the isospectrality is to
require that $\dot b=0$, which requires gauging away $-\alpha(\lambda)$.
This is justified on general grounds by observing that for any Lax equation
$\dot L=[B, L]$, $B$ is not uniquely defined.  In particular, any
term commuting with $L$ can be added to $B$ without changing
Lax equations.  In our case, we are adding a multiple of the identity
to the original formulation in \cite{song2011new}.  Furthermore,
the gauge transformation leading up to \eqref{eq:xLax} is not unimodular which takes us outside of the $SL_2$ theory.
\end{remark}

The polynomials $p_k, q_k$ can be constructed integrating  directly the initial value problem
\begin{equation}\label{eq:xLaxIVP}
\Phi_x=\begin{bmatrix}0 &h\\
-z g& 0 \end{bmatrix} \Phi, \qquad \Phi_1(-\infty)=0, \quad \Phi_2(-\infty)=1,
\end{equation}
with the same rule regarding the multiplication of discrete measures $g,h$
by piecewise smooth left-continuous functions as specified for
\eqref{eq:xLaxBVP}.

  With this convention in place we obtain the following
 characterization of $\Phi_1(x)$ and $\Phi_2(x)$, proven in its entirety in  \cite[Lemma 2.4]{chang-szmigielski-m1CHlong}.
 \begin{lemma} \label{lem:seriessol-}  Let us set
 \begin{equation*}
             \Phi_1(x)=\sum\limits_{0\leq k} \Phi_1^{(k)}(x) z^k, \quad \qquad \Phi_2(x)=            \sum\limits_{0\leq k} \Phi_2^{(k)}(x) z^k.
 \end{equation*}
Then
\begin{equation*}
    \Phi_1^{(0)}(x)=\int\limits_{\eta_0 <x}h(\eta_0) d\eta_0, \quad \qquad \Phi_2^{(0)}(x)=1
\end{equation*}
~for $k=0$, otherwise
\begin{subequations}
\begin{align}
\Phi_1^{(k)}(x)&=(-1)^k \int\limits_{\eta_0<\xi_1<\eta_1<\dots<\xi_k<\eta_k<x}
\big[\prod_{p=1}^k h(\eta_p)g(\xi_p)\big]   h(\eta_0)\, \,  d\eta_0 d\xi_1\dots d\eta_k,   &\\
\Phi_2^{(k)}(x)&=(-1)^k \int\limits_{\xi_1<\eta_1<\dots<\xi_k<\eta_k<x}
\big[\prod_{p=1}^k g(\eta_p)h(\xi_p)\big] \,\,  d\xi_1\dots d\eta_k.   &
\end{align}
\end{subequations}
If the points of the support of  the discrete measure $g$ (and $h$) are ordered $x_1<x_2<\dots<x_N$
then
 \begin{subequations}
\begin{align}
\Phi_1^{(k)}(x)&=(-1)^k \sum_{\substack{j_0<i_1<j_1<\dots< i_k< j_k\\x_{j_k}<x}}
\, \big[\prod_{p=1}^k h_{j_p}g_{i_p}\big] \,\,  h_{j_0} ,  \,  &\\
\Phi_2^{(k)}(x)&=(-1)^k\sum_{\substack{i_1< j_1<\dots< i_k< j_k\\x_{j_k}<x}}
\, \big[\prod_{p=1}^k g_{j_p}h_{i_p}\big] \,\,   .   \, &
\end{align}
\end{subequations}
\end{lemma}

To simplify the formulas in Lemma \ref{lem:seriessol-} we introduce the following notation.
Our basic set of indices is $\{1,2, \dots, N\}$ which we denote by
$[N]$ and if $k\leq N$ we set $[k] = \{ 1,2,\dots,k \}$.
We will denote by  capital letters $I$  and~$J$ any subsets of these sets and use the notation
$\binom{[k]}{j}$ for the set of all $j$-element subsets of $[k]$, listed in increasing order; for example $I\in \binom{[k]}{j}$ means that
$I=\{i_1, i_2,\dots, i_j\}$ for some increasing sequence $i_1 < i_2 < \dots < i_j\leq ~k$.
Furthermore, given a multi-index $I=\{i_1, i_2,\dots, i_j\}$ and a set of
numbers $a_{i_1}, \dots, a_{i_j}$ indexed by $I$, we will abbreviate $a_I=a_{i_1}a_{i_2}\dots a_{i_j}$ etc.

\begin{definition}\label{def:bigIndi} Let $I,J \in \binom{[k]}{l}$, or $I\in \binom{[k]}{l+1},J \in \binom{[k]}{l}$.
\mbox{}

Then  $I, J$ are said to be \emph{interlacing}, denoted $I<J$, if
\begin{equation*}
  \label{eq:interlacing}
    i_{1} <j_{1} < i_{2} < j_{2} < \dotsb < i_{l} <j_{l}
\end{equation*}
or,
\begin{equation*}
    i_{1} <j_{1} < i_{2} < j_{2} < \dotsb < i_{l} <j_{l}<i_{l+1},
\end{equation*}
in the latter case.
The same notation is used
 in the degenerate case $I\in \binom{[k]}{1}, J \in \binom{[k]}{0}$.
\end{definition}
Using this notation we can now express the results of Lemma \ref{lem:seriessol-} in a compact form.
\begin{corollary}\label{cor:solqkpk}
The unique solutions $q_k$ and $p_k$ to  the recurrence equations
\eqref{dstring}  with initial conditions $q_0=0, p_0=1$ are given by
\begin{subequations}
\begin{align}
q_k(z)&=
\sum_{l=0}^{\lfloor\frac{k-1}{2}\rfloor}\Big(\sum_{\substack{I\in \binom{[k]}{l+1}, J\in \binom{[k]}{l}\\ I<J}}
\, h_Ig_J\Big) (-z)^l,  \\
p_{k}(z)&=1+\sum_{l=1}^{\lfloor\frac{k}{2}\rfloor}\Big(\sum_{\substack{I,J \in \binom{[k]}{l}\\ I<J}} h_I g_J
\, \Big)(-z)^l.
\end{align}
\end{subequations}
\end{corollary}

We can now make a brief comment about the spectrum of the boundary value problem
\eqref{eq:xLaxBVP}. We observe that
a complex number z is an \textit{eigenvalue}  of the boundary value problem
\eqref{eq:xLaxBVP} if there exists a solution $\{q_k(z), p_k(z)\}$ to \eqref{dstring} for
which $p_{N}(z)=0$.   The set of all eigenvalues comprises the \textit{spectrum}  of the boundary value problem \eqref{eq:xLaxBVP}.
Our choice of boundary conditions was picked to ensure the
invariance of the spectrum under the time evolution.
To verify that the flow is isospectral (spectrum preserving) we
 examine the $t$ part of the Lax pair \eqref{2chlax} in the region $x>x_{N}$,
  as indicated in Remark \ref{rem:V} and perform the gauge transformation
  to determine the flow of $\Phi$.
\begin{lemma}\label{lem:t-evolution of qp}
Let $\{q_k, p_k\}$ satisfy the system of difference equations \eqref{dstring}.  Then the Lax equations \eqref{2chlax} imply
\begin{equation}\label{eq:tderqp}
  \dot q_{N}=\frac{2}{z}q_{N}-\frac{2u_+}{z}\,p_{N}, \qquad \dot p_{N}=0,
\end{equation}
where $u_+=\sum_{j=1}^{N}h_j$.
\end{lemma}
This lemma implies
that the polynomial $p_{N}(z)$ is independent of time and, in particular, its zeros, i.e. the spectrum, are time invariant. Furthermore, Corollary \ref{cor:solqkpk} allows one to write
the coefficients of $p_N(z)$ in terms of the variables $g_j, h_j$ (or  equivalently $m_j, n_j, x_j$) and thus identify $\lfloor\frac{N}{2}\rfloor$ constants of motion of the system \eqref{eq:m2CH_ode}:
\begin{equation} \label{eq:Mj}
M_j=\sum_{\substack{I,J \in \binom{[N]}{j}\\ I<J}} h_I g_J, \qquad 1\leq j\leq \lfloor\frac{N}{2}\rfloor.
\end{equation}

In the next section we will investigate the role of these constants in the
integrability of \eqref{eq:m2CH_ode}.
\section{Liouville integrability}
\subsection{Bi-Hamiltonian structure}
The results of the previous section, especially
the existence of $\lfloor\frac{N}{2}\rfloor$ constants, suggests that the system \eqref{eq:m2CH_ode} might be
integrable in a classical Liouville sense which is proven below.
For smooth solutions $u(x,t), v(x,t)$ of \eqref{eq:m2CH} the Hamiltonian structure, in fact a bi-Hamiltionian one, of the 2-mCH equation \eqref{eq:m2CH} was given by Tian and Liu in \cite{tian2013tri}. By employing two compatible Hamiltonian
operators
\begin{eqnarray*}
\mathscr{L}_1=
\left(
\begin{array}{cc}
D_xmD_x^{-1}mD_x&D_xmD_x^{-1}nD_x\\
D_xnD_x^{-1}mD_x&D_xnD_x^{-1}nD_x
\end{array}
\right),
\ \mathscr{L}_2=
\left(
\begin{array}{cc}
0&-D_x^2-D_x\\
D_x^2-D_x&0
\end{array}
\right)
\end{eqnarray*}
and the Hamiltonians
\begin{eqnarray}
H_1=\int n (u_x-u) dx,\qquad H_2=\frac{1}{2}\int n(v+v_x)(u-u_x)^2dx,
\end{eqnarray}
the 2-mCH equation \eqref{eq:m2CH}  can be written as
\begin{eqnarray}
\left(
\begin{array}{c}
m_t\\
\ \\
n_t
\end{array}
\right)=\mathscr{L}_1\left(
\begin{array}{c}
\frac{\delta H_1}{\delta m}\\
\ \\
\frac{\delta H_1}{\delta n}
\end{array}
\right)=\mathscr{L}_2\left(
\begin{array}{c}
\frac{\delta H_2}{\delta m}\\
\ \\
\frac{\delta H_2}{\delta n}
\end{array}
\right).
\end{eqnarray}
We note that the word \textit{ compatible} mentioned above means that an arbitrary linear
combination of the two Hamiltonian operators is also Hamiltonian.
Since we work in the non-smooth context the results obtained
for smooth functions will not hold in the non-smooth region, and one either has to formulate a limiting procedure leading to the non-smooth sector or
study the non-smooth sector independently.  At present, we prefer the
second approach mainly because it is technically simpler, and also because it is not clear at this point which Hamiltonian structures have meaningful limits.
\subsection{Hamiltonian vector field}
We focus on the peakon sector of \eqref{eq:m2CH} described by the
system of equations \eqref{eq:m2CH_ode}.
\begin{theorem}\label{thm:hamiltonianODEs}
The equations \eqref{eq:m2CH_ode} for the motion of N peakons of the original PDE \eqref{eq:m2CH} are given by Hamilton's equations of motion:
\begin{align}
\dot x_j=\{x_j,H\},\qquad \dot m_j=\{m_j,H\}, \qquad \dot n_j=\{n_j,H\},
\end{align}
for the Hamiltonian
\begin{equation*} 
H=-\frac{1}{2}\int n(\xi)(u_\xi(\xi)-u(\xi)) d\xi=2M_1+\sum_{k=1}^Nm_kn_k. 
\end{equation*} 
Here
$M_1$ is a constant of motion appearing in \eqref{eq:Mj}, the Poisson bracket $\{,\}$ is given by
\begin{subequations}\label{eq:Poisson bracket}
\begin{align}
&\{x_i,x_k\}=\sgn(x_i-x_k), \\
&\{m_i,m_k\}=\{m_i,x_k\}=\{n_i,n_k\}=\{n_i,x_k\}=\{n_i,m_k\}=0,
\end{align}
\end{subequations}
and the ordering condition $x_1<x_2<\cdots<x_N$ is in place.
\end{theorem}
\begin{proof}
Clearly,
$$\{m_j,h\}=\{n_j,h\}=0$$
under the above Poisson bracket, hence
$$\dot m_j=\dot n_j=0.$$
We proceed with the computation of $ \{x_j,H\}$:
\begin{align*}
 \{x_j,H\}&=\left\{x_j,\ 2\sum_{1\leq i<k\leq N}m_in_ke^{x_i-x_k}+\sum_{k=1}^Nm_kn_k\right\}\\
 &=2\sum_{1\leq i<k\leq N}m_in_k\left\{x_j,e^{x_i-x_k}\right\}\\
 &=2\sum_{1\leq i<k\leq N}m_in_k e^{x_i-x_k}\left(\sgn(x_j-x_i)-\sgn(x_j-x_k)\right)\\
 &=2\sum_{k=1}^{j-1}m_kn_je^{x_k-x_j}+2\sum_{k=j+1}^Nm_jn_ke^{x_j-x_k}+4\sum_{1\leq i<j<k\leq N}m_in_ke^{x_i-x_k}\\
 &\stackrel{\eqref{eq:Q}}{=}\langle Q \rangle(x_j),
\end{align*}
thus proving the results.
\end{proof}
\subsection{Liouville integrability}
We will introduce a natural Poisson manifold $(M, \pi)$ defined
by the Poisson bracket \eqref{eq:Poisson bracket}.  Since $m_j, n_j$
are constant we can restrict our considerations to the non-trivial
part of the Poisson structure involving only $x_j$.
Let us denote
\begin{equation}
M=\big\{x_1<x_2<\cdots<x_N\big\}
\end{equation}
 and define
 \begin{equation}
 \pi(f,g)=\{f,g\}=\sum_{1\leq i<j\leq N} \{x_i, x_j\} \frac{\partial{f}}{\partial{x_i}}\frac{\partial{g}}{\partial{x_j}}
\end{equation}
 for all differentiable functions $f,g$ on $M$.  Then $M$ has a structure of a  Poisson
 manifold $M$ to be denoted
 $(M, \pi)$.

One can check directly from \eqref{eq:Poisson bracket} that regardless
whether $N=2K$ or $N=2K+1$
\begin{lemma} \label{lem:rank pi}
 $rank(\pi)=2K$.
\end{lemma}
Our objective now is to identify an appropriate number of
Poisson commuting quantities.  We will break down our analysis
according to whether $N$ is even or odd.
\begin{enumerate}
\item[1.] \textbf{Case $N=2K$.}
It follows from  Lemma \ref{lem:t-evolution of qp} and \eqref{eq:Mj} that the quantities
\begin{equation} \label{eq:Hjs-even}
M_j=\sum_{\substack{I,J \in \binom{[2K]}{j}\\ I<J}} h_I g_J, \qquad 1\leq j\leq K,
\end{equation}
with $ h_i=m_ie^{x_i},\ \  g_i=n_ie^{-x_i} $, form a set of $K$  constants  of motion for the system \eqref{eq:m2CH_ode}.
We claim that these constants of motion Poisson commute.
Short of giving a detailed proof, we would like to outline the argument
which goes back to J. Moser in  \cite{moser-three}.  Since $M_j$ commute with the Hamiltonian $H$ (see Theorem \ref{thm:hamiltonianODEs}) their
Poisson bracket $\{M_j, M_k\}$ commutes with $H$ and thus
$\{M_j, M_k\}$ is a constant of motion for every pair of indices $j, k$.
For cases for which the inverse spectral methods allow one to express
$M_j$ in terms of leading asymptotic positions, in particular
exploiting the asymptotic result that particles corresponding to
to adjacent positions $x_j, x_{j+1}$ pair up, while distinct pairs do not
interact, leads to a suppression of the majority of terms in $M_j$.
The precise argument is presented in \cite[Theorem 3.8]{chang-szmigielski-liouville1mCH}  while needed asymptotic results can be found in Theorem \ref{thm:evenass}.
\begin{theorem} \label{thm:even commute}
The Hamiltonians  $M_1,\cdots,M_K$  Poisson commute.
\end{theorem}

\item[2.] \textbf{Case $N=2K+1$.}
Again, following Lemma \ref{lem:t-evolution of qp} and \eqref{eq:Mj} we see
\begin{equation} \label{eq:Hjs-odd}
M_j=\sum_{\substack{I,J \in \binom{[2K+1]}{j}\\ I<J}} h_I g_J, \qquad 1\leq j\leq K,
\end{equation}
with $h_i=m_ie^{x_i},\ \ \ g_i=n_ie^{-x_i}$,
are constants  of motion for the system \eqref{eq:m2CH_ode} in the odd case.

In the odd case, there is an extra constant of motion,
which can be computed from the value of the Weyl function $W(z)$ at $z=\infty$ (see Section \ref{sec:Lax} and Section \ref{sec:FSM} for details regarding the Weyl function). We point out that this constant is $0$ in the even case.    The computation is routine and produces
\begin{equation} \label{eq:c}
c=\frac{\sum\limits_{\substack{I\in \binom{[2K+1]}{K+1},J \in \binom{[2K+1]}{K}\\ I<J}} h_I g_J}{\sum\limits_{\substack{I,J \in \binom{[2K+1]}{K}\\ I<J}} h_I g_J}=\frac{\sum\limits_{\substack{I\in \binom{[2K+1]}{K+1},J \in \binom{[2K+1]}{K}\\ I<J}} h_I g_J}{M_K},
\end{equation}
which, in turn, gives an extra constant of motion
$$M_c=\sum\limits_{\substack{I\in \binom{[2K+1]}{K+1},J \in \binom{[2K+1]}{K}\\ I<J}} h_I g_J=\prod_{j=1}^{K+1}m_{2j-1}e^{x_j}\prod_{j=1}^{K}n_{2j}e^{-x_j},$$
so that $\{M_1,M_2,\cdots, M_K,M_c\}$ form a set of $K+1$  constants  of motion for the system \eqref{eq:m2CH_ode} in this case.

It is not hard to see by using the same argument as  in \cite[Theorem 3.9]{chang-szmigielski-liouville1mCH}  and the asymptotic results in Theorem \ref{thm:oddass}, that the following theorem holds.
\begin{theorem}\label{thm:odd commute}
The Hamiltonians $M_1,\cdots,M_K,M_c$ Poisson commute.
\end{theorem}
\end{enumerate}

Combining now both theorems above we conclude (the proof is similar to the argument in  \cite[Theorem 3.10]{chang-szmigielski-liouville1mCH}).
\begin{theorem} \label{thm:Liouville integrability}
The conservative peakon system given by \eqref{eq:m2CH_ode}
is Liouville integrable.
\end{theorem}

\section{Forward map: spectrum and spectral data} \label{sec:FSM}
The spectrum of the boundary value problem \eqref{eq:xLaxBVP} (or equivalently, \eqref{dstring}) is
given by the zeros of the polynomial $p_N(z)$.  However,
one cannot recover the measures $g$ and $h$ from the spectrum alone.
One needs extra data and the right object to turn to is
the \textit{Weyl function}
\begin{equation}\label{eq:defWeyl}
W(z)=\frac{q_{N}(z)}{p_{N}(z)},
\end{equation}
which in our case is a rational function with poles located at the
spectrum of the boundary value problem.  Another compelling reason for
using the Weyl function is that, as we will show below, the residues
of $W$ evolve linearly in time, while the value of $W$ at $z=\infty$ is
a constant of motion.
The investigation of the analytic properties of $W$ can be greatly
simplified by observing that $W$ is built out of solutions to the
recurrence \eqref{dstring}.  This suggests forming a recurrence
of Weyl functions whose solution at step $N$ is $W(z)$.
This leads to the following result which is an immediate consequence of
\eqref{dstring}.

 \begin{lemma}\label{lem:wdstring}
 Let $\{q_k, p_k\}$ be the solution to \eqref{dstring} and let $w_{2k}=\frac{q_{k}}{p_{k}}, w_{2k-1}=\frac{q_{k-1}}{p_{k}}$.
 Then
 \begin{subequations}
 \begin{align}
    w_1=0, \qquad w_{2k}&=(1+zm_kn_k)w_{2k-1}+h_k, \quad &&1\leq k\leq  N,  \label{eq:recweven}\\
    \qquad \frac{1}{w_{2k}}&=\frac{1}{w_{2k+1}}+zg_{k+1},  \quad &&1\leq k\leq N-1.   \label{eq:recwodd}
\end{align}
\end{subequations}
  \end{lemma}

  We will now show that all these Weyl functions,
  including the original $W(z)$,  have the following properties in
  common:
  \begin{enumerate}
  \item they all have simple poles located on $\R_+$;
  \item all the residues are positive;
  \item the values at $z=\infty$ are non-negative.
  \end{enumerate}
  The rational functions of this type have been studied, as a special case,  in the famous
  memoir by T. Stieltjes \cite{stieltjes}.  The most relevant for our studies
  is the following theorem which is a special case of
  a more general theorem proved by Stieltjes.
  \begin{theorem} [T. Stieltjes] \label{thm:Stieltjes} Any rational function $F(z)$ admitting the
integral representation
\begin{equation}\label{eq:Stieltjesintegral}
F(z)=c+\int \frac{d\nu(x)}{x-z},
\end{equation}
where $d\nu(x)$ is the (Stieltjes) measure corresponding to the piecewise constant
non-decreasing function $\nu(x)$ with finitely many jumps in $\R_+$
has a finite (terminating) continued fraction expansion
\begin{equation}\label{eq:Stieltjescf}
F(z)=c+\cfrac{1}{a_1 (-z)+\cfrac{1}{a_2+\cfrac{1}{a_3(-z)+\cfrac{1}{\ddots}}}},
\end{equation}
where all $a_j>0$ and, conversely, any rational function with this type of
a continued fraction expansion has the integral representation \eqref{eq:Stieltjesintegral}.
\end{theorem}
We now apply Stieltjes' result to our case.

  \begin{lemma}\label{lem:spectralwks}
Given $h_j, g_j>0, h_jg_j=m_jn_j>0, 1\leq j\leq N$, let $w_j$s satisfy the recurrence relations of Lemma \ref{lem:wdstring}.
Then $w_j$s are shifted Stieltjes transforms of finite, discrete Stieltjes measures supported on $\R_+$, with nonnegative shifts.  More precisely:
\begin{align*}
w_{2k-1}(z)&=\int \frac{d\mu^{(2k-1)}(x)}{x-z}, \\
w_{2k}(z)&=c_{2k}+\int \frac{d\mu^{(2k)}(x)}{x-z},
\end{align*}
where $c_{2k}>0$ when $k$ is odd, otherwise, $c_{2k}=0$.  Furthermore,
the number of points in the support $d\mu^{(2k)}(x)$ and $d\mu^{(2k-1)}$ is
$\lfloor\frac{k}{2}\rfloor$.
\end{lemma}
\begin{proof}
We only sketch the proof,  for further details we refer to \cite[Lemma 3.6]{chang-szmigielski-m1CHlong}.  The proof goes by induction on $k$.
The base case $k=1$ is elementary.  Assuming the induction
hypothesis to hold up to $2k$ we invert \eqref{eq:recwodd} to get:
\begin{equation*}
w_{2k+1}(z)=\frac{1}{-zg_{k+1}+\frac{1}{w_{2k}}}
\end{equation*}
which, by induction hypothesis, implies that $w_{2k+1}$ has the required
continued fraction expansion covered by Stieltjes' theorem, and thus has
the required integral representation.
We subsequently feed this integral representation into \eqref{eq:recweven}
to obtain the Stietljes integral representation for $w_{2k+2}$.
The analysis of the signs of the values of the Weyl functions at $z=\infty$ is carried out in \cite[Lemma 3.6]{chang-szmigielski-m1CHlong}.

\end{proof}

\begin{remark} The recurrence in Lemma \ref{lem:wdstring} can be viewed as
the recurrence on the Weyl functions corresponding to shorter bars (keeping in mind the interpretation in terms of the longitudinal vibrations
of an elastic bar)
obtained by truncating at the index $k$.  Then $W_{2k}$ is precisely the
Weyl function corresponding to the measures $\sum_{j=1}^k h_j \delta_{x_j}$ and $\sum_{j=1}^k g_j \delta_{x_j}$,
while $W_{2k-1}$ corresponds to the measures $\sum_{j=1}^{k-1}  h_j \delta_{x_j}$ and $\sum_{j=1}^k g_j \delta_{x_j}$ respectively.
\end{remark}

Now, in particular, we note that by Lemma \ref{lem:spectralwks}
\begin{equation}\label{eq:WrepS}
W(z)=\frac{q_{N}(z)}{p_{N}(z)}=c_{2N}+\int \frac{d\mu^{(2N)}(x)}{x-z}, \qquad d\mu^{(2N)}=\sum_{j=1}^{\lfloor \frac N2 \rfloor} b_j^{(2N)}\delta_{\zeta_j}
\end{equation}
and thus the following theorem holds.
\begin{theorem} \label{thm:W}
$W(z)$ is a (shifted) Stieltjes transform of a positive, discrete measure $d\mu$ with support
inside $\R_+$.  More precisely:
\begin{equation*}
W(z)=c+\int \frac{d\mu(x)}{x-z}, \quad d\mu=\sum_{i=1}^{\lfloor \frac N2 \rfloor}
b_j \delta_{\zeta _j}, \  0<\zeta_1<\dots< \zeta_{\lfloor \frac N2 \rfloor}, \ \ 0<b_j,   \  1\leq j\leq \floor{\frac N2},
\end{equation*}
where
$c>0$ when $N$ is odd
and $c=0$ when $N$ is even.
\end{theorem}
The next corollary summarizes the properties of the spectrum of the boundary value problem \eqref{eq:xLaxBVP}, or equivalently \eqref{dstring}.
\begin{corollary} \label{cor:spectrum}
\mbox{}
\begin{enumerate}
\item The spectrum of the boundary value problem \eqref{eq:xLax} is positive and simple.
\item
$W(z)=c+\sum_{j=1}^{\lfloor \frac N2 \rfloor} \frac{b_j}{\zeta_j-z}$, where
 all residues satisfy $b_j>0$ and $c\geq0$.
 \end{enumerate}
 \end{corollary}

\section{Inverse problem}\label{sec:IP}

\subsection{The first inverse problem and the interpolation problem}

The initial inverse problem we are interested
in solving can be stated as follows:
\begin{definition} \label{def:ISP}
Given a rational function (see Theorem \ref{thm:W})
\begin{equation} \label{eq:Wc}
W(z)=c+\int \frac{d\mu(x)}{x-z}, \quad d\mu=\sum_{i=1}^{\lfloor \frac N2 \rfloor}
b_j \delta_{\zeta _j}, \  0<\zeta_1<\dots< \zeta_{\lfloor \frac N2 \rfloor}, \ \ 0<b_j,   \  1\leq j\leq \floor{\frac N2},
\end{equation}
where
$c>0$ when $N$ is odd
and $c=0$ when $N$ even, as well as
positive constants $m_1,m_2,\dots, m_N,n_1,n_2,\cdots,n_N$, such that the products $m_jn_j$ are distinct,
find positive constants $g_j, h_j$, $1\leq j\leq N$, for which
\begin{enumerate}
\item $g_jh_j=m_jn_j$ ,
\item the unique solution of the initial value problem:
\begin{align} \label{eq:rodIVP}
&\begin{bmatrix} q_k\\ p_k \end{bmatrix}=\begin{bmatrix} 1& h_k\\ -z g_k&1 \end{bmatrix}\begin{bmatrix} q_{k-1}\\ p_{k-1} \end{bmatrix} , \hspace{1cm} \qquad \qquad 1\leq k\leq N, \\
 &\begin{bmatrix} q_0\\ p_0\end{bmatrix}=\begin{bmatrix} 0 \notag\\ 1\end{bmatrix},
      \end{align}
satisfies
\begin{equation*}
\frac{q_N(z)}{p_N(z)}=W(z).
\end{equation*}
\end{enumerate}
\end{definition}

\begin{remark} The restriction that the products $m_jn_j$ be distinct has been made to
facilitate the argument and
will be eventually relaxed by taking appropriate limits of the generic case (see the comments below
Theorem \ref{thm:inversex}).
\end{remark}

The basic idea of the solution of the above inverse problem is to
associate to it a certain interpolation problem.  We now proceed to explain
how such an interpolation problem appears already in the
solution of the forward problem, that is, in the solution of the difference
boundary value problem \eqref{dstring}, or equivalently
\eqref{eq:rodIVP} above.
First, let us denote by
\begin{equation}\label{eq:Tk}
T_k(z)=\begin{bmatrix} 1&h_k \\- zg_k&1 \end{bmatrix}, 
\end{equation}
the transition matrix appearing in \eqref{eq:rodIVP}.
Clearly,
\begin{equation}\label{eq:iteration}
\begin{bmatrix} q_N(z) \\ p_N(z) \end{bmatrix}=T_N(z) T_{N-1}(z) \cdots T_{N-k+1} (z) \begin{bmatrix} q_{N-k}(z)\\ p_{N-k}(z) \end{bmatrix}.
\end{equation}
Let us introduce a different indexing $i'=N-i+1$ which is a bijection of the
set $[1,N]$ and represents counting
points of the beam from right to left rather than left to right.  Moreover, let us denote by
$\hat T_j(z)$ the classical adjoint of $T_{j'}$.  Thus
\begin{equation}
\hat T_j(z)=\begin{bmatrix} 1&-h_{N-j+1}\\zg_{N-j+1}&1 \end{bmatrix} =
\begin{bmatrix} 1&-h_{j'}\\zg_{j'}&1 \end{bmatrix}.
\end{equation}
Then \eqref{eq:iteration} implies
\begin{equation}\label{eq:backiteration}
\hat T_k(z) \cdots\hat T_{2} \hat T_{1} (z)\begin{bmatrix} W(z)\\ 1 \end{bmatrix}= \frac{\prod_{j=1}^k \big(1+zm_{j'}n_{j'}\big)}{p_N(z)} \begin{bmatrix} q_{(k+1)'}(z)\\ p_{(k+1)'}(z) \end{bmatrix},
\end{equation}
where we used that $\det(\hat T_{j'})=1+zg_{j'}h_{j'}$ and, subsequently, $g_{j'}h_{j'}=m_{j'}n_{j'}$.  We clearly have
\begin{lemma}
Let us fix $1\leq k\leq N$ and denote $\hat S_k(z)=\hat T_k(z) \cdots\hat T_{2} \hat T_{1} (z)$.  Then for every  $1\leq j\leq k$ the vector $\begin{bmatrix} W(z_j)\\ 1 \end{bmatrix}$, where $z_j=-\frac{1}{m_{j'}n_{j'}}$,
is in the $\textrm{ker}\big(\hat S_k(z_j)\big)$.
\end{lemma}
We proceed by explicitly writing the conditions that
the vectors $\begin{bmatrix} W(z_i)\\ 1 \end{bmatrix}$ be null vectors
of $\hat S_k(z_i)$.  To this end we
write
\begin{equation}\label{eq:Skhat}
\hat S_k(z)=\begin{bmatrix} \hat q_k(z)&\hat Q_k(z)\\ \hat p_k(z)&\hat P_k(z) \end{bmatrix},
\end{equation}
from which two sets of interpolation conditions
\begin{subequations} \label{eq:interpol}
\begin{align}
\hat q_k(z_j)W(z_j)+ \hat Q_k(z_j)=0, \quad 1\leq j\leq k, \\
\hat p_k(z_j)W(z_j)+ \hat P_k(z_j)=0, \quad 1\leq j\leq k,
\end{align}
\end{subequations}
emerge.
Whether these conditions, given $W(z)$, determine the polynomials $\hat q_k, \hat Q_k$, $ \hat p_k, \hat P_k$
will depend on their degrees and this is the subject of the next result whose proof is an easy exercise in induction.
\begin{lemma}\label{lem:degrees} For any $k, \,  1\leq k\leq N, $
\begin{enumerate}
\item $\deg \hat q_k(z)=\lfloor \tfrac{k}{2}\rfloor, \, \det \hat Q_k(z)= \lfloor \tfrac{k-1}{2}\rfloor, \, \deg \hat p_k(z)=\lfloor \tfrac{k+1}{2}\rfloor, \, \det \hat P_k(z)= \lfloor \tfrac{k}{2}\rfloor$,
\item $\hat q_k(0)=1,\,  \hat p_k(0)=0, \, \hat P_k(0)=1. $
\end{enumerate}
\end{lemma}
Now it is elementary to check that the number of interpolation conditions
in equations \eqref{eq:interpol} is the same as the number of unknown coefficients
in $\hat q_k, \hat Q_k, \hat p_k, \hat Q_k$, so, in principle, the
solution exists.
Before we state our next lemma we revisit the notation introduced
in Definition \ref{def:bigIndi}.  For any multi-index $I\in \binom{[k]}{j}$,
where we recall $I=\{i_1, i_2, \cdots,  i_j\}$ is an ordered set associated
to an increasing sequence $1\leq i_1<i_2\cdots<i_j\leq k$, we
assign its ordered image $I'$ obtained by applying the bijection $i\rightarrow N+1-i$ to $I$ and reordering.  The following result
can be demonstrated by using induction on $k$ and the definition of
$\hat S_k(z)$ \eqref{eq:Skhat}.
\begin{lemma} For any $k$, $1\leq k\leq N$,
\begin{subequations}
\begin{align}
&\hat q_{k}(z)=1+\sum_{j=1}^{\lfloor\frac{k}{2}\rfloor}\Big(\sum_{\substack{I,J \in \binom{[k]}{j}\\ I<J}} g_{I'} h_{J'}
\, \Big)(-z)^j, \\
&\hat Q_{k}(z)=-\sum_{j=0}^{\lfloor\frac{k-1}{2}\rfloor}\Big(\sum_{\substack{I\in \binom{[k]}{j+1}, J\in \binom{[k]}{j}\\ I<J}} h_{I'}g_{J'}
\, \Big)(-z)^j, \\
&\hat p_k(z)=
-\sum_{j=1}^{\lfloor\frac{k+1}{2}\rfloor}\Big(\sum_{\substack{I\in \binom{[k]}{j}, J\in \binom{[k]}{j-1}\\ I<J}}
\, g_{I'} h_{J'} \Big) (-z)^j, \\
&\hat P_k(z)=1+
\sum_{j=1}^{\lfloor\frac{k}{2}\rfloor}\Big(\sum_{\substack{I, J\in \binom{[k]}{j}\\ I<J}}
h_{I'} g_{J'} \Big) (-z)^j,
\end{align}
\end{subequations}
with the convention that $\hat q_1(z)=1,\,  \hat Q_1(z)=-h_N, \, \hat p_1(z)=zg_N, \, \hat P_1(z)=1$.
\end{lemma}

\begin{example}
It is instructive to display $\hat S_k(z)$ for small $k$.
The notation is that of \eqref{eq:Skhat}.

\begin{enumerate} [label=\alph*)]
\item
\begin{equation*}
\hat S_1(z)=\begin{bmatrix} \hat q_1(z)& \hat Q_1(z)\\\hat p_1(z)&\hat P_1(z) \end{bmatrix}=\begin{bmatrix} 1 &-h_{1'}\\ g_{1'}z& 1 \end{bmatrix}.
\end{equation*}
\item
\begin{equation*}
\hat S_2(z)=\begin{bmatrix} \hat q_2(z)& \hat Q_2(z)\\\hat p_2(z)&\hat P_2(z) \end{bmatrix}=\begin{bmatrix} 1 -(g_{1'} h_{2'})z&-(h_{1'}+h_{2'})\\ (g_{1'}+g_{2'})z& 1 -(h_{1'}g_{2'})z\end{bmatrix}.
\end{equation*}
\item
\begin{align*}
&\hat S_3(z)=\begin{bmatrix} \hat q_3(z)& \hat Q_3(z)\\\hat p_3(z)&\hat P_3(z) \end{bmatrix}=\\
&\quad \begin{bmatrix} 1 -(g_{1'} h_{2'}+g_{1'}h_{3'}+g_{2'}h_{3'})z&-(h_{1'}+h_{2'}+h_{3'})+(h_{1'}g_{2'} h_{3'}) z\\ (g_{1'}+g_{2'}+g_{3'})z-(g_{1'}h_{2'} g_{3'})z^2& 1 -(h_{1'}g_{2'}+h_{1'}g_{3'}+h_{2'}g_{3'})z\end{bmatrix}. 
\end{align*}
\end{enumerate}
\end{example} 
For a polynomial $f(z)$ let us denote by $f^+$ the coefficient at the
highest power in $z$ and use the convention $\hat q_0=1$.
Then, by inspection, we obtain
\begin{equation*}
g_{1'}=\frac{\hat p_{1}^+}{\hat q_0^+}, \quad g_{2'}=\frac{\hat P_{2}^+}{\hat Q_1^+}, \quad g_{3'}=\frac{\hat p_{3}^+}{\hat q_2^+},
\end{equation*}
 and continuing with the help of induction we are led to
\begin{theorem} \label{thm:invg}
For any $1\leq k\leq N$,
\begin{subequations}
\begin{align}
g_{k'}&=\frac{\hoc{\hat p_k}}{\hoc{\hat q_{k-1}}},   &\text{if k  is odd,} \label{eq:inversegodd}\\
g_{k'}&=\frac{\hoc{\hat P_k}}{\hoc{\hat Q_{k-1}}}, &\text{if k is even.}\label{eq:inversegeven}
\end{align}
\end{subequations}
\end{theorem}
\begin{remark}

Since, $g_{k'}h_{k'}=m_{k'}n_{k'}$, knowing $g_{k'}$ determines uniquely
$h_{k'}$.
\end{remark}
Now we can state our strategy for solving the original inverse problem
stated in Definition \ref{def:ISP}:
\begin{enumerate}
\item given $W(z)$ we solve the interpolation problems  \eqref{eq:interpol}
for all $1\leq k\leq N$;
\item from the solution to the interpolation problem at stage
$1\leq k\leq N$, we use Theorem \ref{thm:invg} to recover $g_{k'}, h_{k'}$ and thus
$\hat T_k =\begin{bmatrix} 1 & -h_{k'}\\ zg_{k'}&1 \end{bmatrix}$, finally
all transitions matrices $T_k$ (see \eqref{eq:Tk});
\item we then define
\begin{equation}\label{eq:forwarditeration}
\begin{bmatrix} q_k(z) \\ p_k(z) \end{bmatrix}=T_k(z) T_{k-1}(z) \cdots T_{1} (z) \begin{bmatrix} 0\\ 1\end{bmatrix};
\end{equation}
\item we define $\tilde W(z)=\frac{q_N(z)}{p_N(z)} $.
The fact that $W(z)=\tilde W(z)$ follows from \eqref{eq:iteration} and \eqref{eq:backiteration} with $k=N$.
\end{enumerate}

The interpolation problem is linear so the solution will be expressed
in terms of determinants.  Before, however, we present the final formulae it is helpful to introduce a family of generalized Cauchy-Vandermonde matrices \cite{gasca1989computation,martinez1998factorizations,martinez1998fast} attached to a Stieltjes transform of a positive measure.  Matrices of this type arise  in some interpolation problems, including the current one.  We refer the reader to \cite{chang-szmigielski-m1CHlong} for more details but
to ease the presentation we provide in the paragraph below a simplified version of the
interpolation problem specified by \eqref{eq:interpol}.

Given three positive integers $k, l, m$ such that $k=l+m$, a function $f(z)$,  and a collection of distinct points $\{z_j\}_1^k$  we are seeking two polynomials $a(z)=1+\sum_{n=1}^l a_n z^n$ and
$b(z)=\sum_{n=0}^m b_n z^n$ such that
\begin{equation*}
a(z_j)f(z_j)+b(z_j)=0, \qquad 1\leq j\leq k.
\end{equation*}
This interpolation problem is equivalent to the matrix problem
\begin{equation*}
\begin{bmatrix} z_1^1f(z_1)&z_1^2f(z_1)&\dots&z_1^lf(z_1)&1&z_1&\dots&z_1^m\\\\
z_2^1f(z_1)&z_2^2f(z_2)&\dots&z_2^lf(z_2)&1&z_2&\dots&z_2^m\\\\
\vdots&\vdots&\vdots&\vdots&\vdots&\vdots&\vdots&\vdots\\\\\\
z_k^1f(z_k)&z_k^2f(z_k)&\dots&z_k^lf(z_k)&1&z_k&\dots&z_k^m \end{bmatrix} \begin{bmatrix} a_1\\ \vdots\\ \vdots\\ a_l\\b_0\\\vdots\\ b_m \end{bmatrix} =-\begin{bmatrix}f(z_1)\\ \vdots\\ \vdots\\ \\ \\ \vdots\\f(z_k) \end{bmatrix}.
\end{equation*}
Denoting the determinant of the matrix on the left by $D_k$ and assuming for now that $D_k\neq 0$
we can succinctly write the solution
\begin{align*}
&a(z)+z^{l+1} b(z)=\\
&\frac{1}{D_k} \det \begin{bmatrix} 1&z
&z^2&\dots&z^l&z^{l+1}&z^{l+2}&\dots&z^k\\
f(z_1)&z_1^1f(z_1)&z_1^2f(z_1)&\dots&z_1^lf(z_1)&1&z_1&\dots&z_1^m\\\\
f(z_2)&z_2^1f(z_1)&z_2^2f(z_2)&\dots&z_2^lf(z_2)&1&z_2&\dots&z_2^m\\\\
\vdots&\vdots&\vdots&\vdots&\vdots&\vdots&\vdots&\vdots&\vdots\\\\\\
f(z_k)&z_k^1f(z_k)&z_k^2f(z_k)&\dots&z_k^lf(z_k)&1&z_k&\dots&z_k^m \end{bmatrix}  .
\end{align*}
In the case of the function $f$ being the Stieltjes transform of a measure
the determinants in question are computable.  Now we turn to spelling out
the most important points in the solution of the inverse problem.  We start with a definition.

~\begin{definition} \label{def:CSV}
 Given  a (strictly) positive vector $\mathbf{e}\in \R^k$, a non-negative
 number $c$, an index $l$ such that $0\leq l\leq k$, another index $p$ such that $0\leq p, \, p+l-1\leq k-l$, and a positive measure $\nu$
 with support in $\R_+$, a \textit{ Cauchy-Stieltjes-Vandermonde (CSV) matrix} is that of the form
\begin{align*}
  &CSV_k^{(l,p)}(\mathbf{e},\nu, c)=\\
  &\left(\begin{array}{cccccccc}
    e_1^{p}\hat\nu_c(e_1)&e_1^{p+1}\hat \nu_c(e_1)&\cdots&e_1^{p+l-1}\hat\nu_c(e_1)&1&e_1&\cdots&e_1^{k-l-1}\\
    e_2^{p}\hat\nu_c(e_2)&e_2^{p+1}\hat\nu_c(e_2)&\cdots&e_2^{p+l-1}\hat\nu_c(e_2)&1&e_2&\cdots&e_2^{k-l-1}\\
    \vdots&\vdots&\ddots&\vdots&\vdots&\vdots&\ddots&\vdots\\
    e_k^{p}\hat\nu_c(e_k)&e_k^{p+1}\hat\nu_c(e_k)&\cdots&e_k^{p+l-1}\hat\nu_c(e_k)&1&e_k&\cdots&e_k^{k-l-1}
  \end{array}\right),
\end{align*}
where  $
\hat \nu_c(y)=c+\int\frac{d\nu(x)}{y+x}
$ is the (shifted) classical Stieltjes transform of the measure
$\nu$.

\end{definition}
\begin{remark}
In this section we use a slightly different definition of the Stieltjes transform
then the one in the context of the Weyl function (see Section \ref{sec:FSM}).
Thus, in this section, it is  $W(-z)$ which is the Stieltjes transform of the
spectral measure; this notation being in fact  more in line with
Stieltjes' notation in \cite{stieltjes}.

\end{remark}

The explicit formulas for the determinant of the CSV matrix can be readily obtained. To this end we need some notations to
facilitate the presentation.
Recall that the
multi-index notation that was introduced earlier in the part leading up to the definition \ref{def:bigIndi}.  Moreover, let us
 denote $[i,j]=\{i,i+1,\cdots,j\}, \,\,  \binom{[1,K]}{k}=\{J=\{j_1,j_2,\cdots,j_k\}|j_1<\cdots <j_k, j_i\in [1,K]\}$.  Then for two ordered multi-index sets $I, J$  we define
\begin{align*}
  \mathbf{x}_J&=\prod_{j\in J}x_j, &\Delta_J(\mathbf{x})&=\prod_{i<j\in J}(x_j-x_i), \\
  \Delta_{I,J}(\mathbf{x};\mathbf{y})&=\prod_{i\in I}\prod_{j\in J}(x_i-y_j),
  &\Gamma_{I,J}(\mathbf{x};\mathbf{y})&=\prod_{i\in I}\prod_{j\in J}(x_i+y_j),
  \end{align*}
  along with the conventions
\begin{align*}
&\Delta_\emptyset(\mathbf{x})=\Delta_{\{i\}}(\mathbf{x})=\Delta_{\emptyset,J}(\mathbf{x};\mathbf{y})=\Delta_{I,\emptyset}(\mathbf{x};\mathbf{y})=\Gamma_{\emptyset,J}(\mathbf{x};\mathbf{y})=\Gamma_{I,\emptyset}(\mathbf{x};\mathbf{y})=1,&\\
&\binom{[1,K]}{0}=1;\qquad\qquad \binom{[1,K]}{k}=0,\ \  k>K.&
\end{align*}

Since we will eventually obtain expressions in terms of the
ratios of determinants it is helpful to study the structure
of these determinants. The results stated below will be important
at two stages of our analysis: the existence of the solution to the inverse problem and the large time asymptotic analysis of solutions to \eqref{eq:m2CH_ode}.
The detailed proofs can be found in  \cite{chang-szmigielski-m1CHlong}.
\begin{theorem}\label{thm:detCSV}
Let $\nu$ be a positive measure with support in $\R_+$ and let $\mathbf{x}$ denote the vector ~$[x_1,x_2,\dots , x_l]\in ~\R^l $ and $d\nu^p(y)=y^p d\nu(y)$, respectively.
Then
\begin{enumerate}
\item if either $c=0$ or $p+l-1<k-l$ then
\begin{equation}\label{eq:detCSV1}
\begin{split}
&\det CSV_k^{(l,p)}(\mathbf{e},\nu, c)=(-1)^{lp+\frac{l(l-1)}{2}}\\
& \Delta_{[1,k]}(\mathbf{e})
\int_{0<x_1<x_2<\dots<x_l} \frac{\Delta_{[1,l]}(\mathbf{x})^2}{\Gamma_{[1,k], [1,l]}(\mathbf{e}; \mathbf{x})} d\nu^p(x_1)d\nu^p(x_2)\dots d\nu^p(x_l);
\end{split}
\end{equation}
\item if $c>0$ and $p+l-1=k-l$ then
\begin{equation} \label{eq:detCSV2}
\begin{split}
&\det CSV_k^{(l,p)}(\mathbf{e},\nu, c)=(-1)^{lp+\frac{l(l-1)}{2}} \Delta_{[1,k]}(\mathbf{e})\\&\cdot\Big(
\int_{0<x_1<x_2<\dots<x_l} \frac{\Delta_{[1,l]}(\mathbf{x})^2}{\Gamma_{[1,k], [1,l]}(\mathbf{e}; \mathbf{x})} d\nu^p(x_1)d\nu^p(x_2)\dots d\nu^p(x_l)\\
&+c
\int_{0<y_1<y_2<\dots<y_{l-1}} \frac{\Delta_{[1,l-1]}(\mathbf{y})^2}{\Gamma_{[1,k], [1,l-1]}(\mathbf{e}; \mathbf{y})} d\nu^p(y_1)d\nu^p(y_2)\dots d\nu^p(y_{l-1})\Big).
\end{split}
\end{equation}
\end{enumerate}

\end{theorem}

Our next step is to give a complete
solution to the inverse problem as stated in Definition \ref{def:ISP} in terms of the determinants of the CSV matrices.
To this end we set (see \eqref{eq:Wc}):
\begin{equation} \label{eq:ej}
e_j=\frac{1}{m_{j'}n_{j'}}, \quad \mathbf{e}_{[1,j]}=e_1e_2\cdots e_j, \qquad \nu =\mu, \quad 1\leq j\leq N,
\end{equation}
and
\begin{equation}\label{eq:calD}
\D_k^{(l,p)}=\abs{\det CSV_k^{(l,p)}(\mathbf{e},\mu,c)}.
\end{equation}

Now, with Theorem \ref{thm:detCSV} in hand, the main theorem of this section follows from the
solution of the interpolation problem \eqref{eq:interpol} with
normalization conditions of Lemma \ref{lem:degrees}, and Theorem \ref{thm:invg}.
For details of computations we refer to \cite{chang-szmigielski-m1CHlong}.

\begin{theorem}\label{thm:detsolISP}
Suppose the Weyl function $W(z)$ is given by \eqref{eq:Wc} along with
positive constants (masses) $m_1, m_2, \cdots, m_N,n_1,n_2,\cdots,n_N$ such that the products $m_jn_j$ are distinct.  Then,
there exists a unique solution to the inverse problem specified in Definition \ref{def:ISP}:
\begin{subequations}
\begin{align}
&&g_{k'}&=\frac{\D_k^{(\frac{k-1}{2},1)}\D_{k-1}^{(\frac{k-1}{2},1)}}
{\mathbf{e}_{[1,k]}\D_k^{(\frac{k+1}{2},0)}\D_{k-1}^{(\frac{k-1}{2},0)}},   &\text{ if k  is odd},  \label{eq:detinversegodd}\\
&&g_{k'}&= \frac{\D_k^{(\frac{k}{2},1)}\D_{k-1}^{(\frac k2 -1,1)}}
{\mathbf{e}_{[1,k]}\D_k^{(\frac k2,0)}\D_{k-1}^{(\frac{k}{2},0)}}, &\text{ if k is even}. \label{eq:detinversegeven}
\end{align}
\end{subequations}
Likewise,
\begin{subequations}
\begin{align}
h_{k'}&=\frac
{\mathbf{e}_{[1,k-1]}\D_k^{(\frac{k+1}{2},0)}\D_{k-1}^{(\frac{k-1}{2},0)}}{\D_k^{(\frac{k-1}{2},1)}\D_{k-1}^{(\frac{k-1}{2},1)}},   &\text{ if k is odd},  \label{eq:detinversehodd}\\
h_{k'}&= \frac
{\mathbf{e}_{[1,k-1]}\D_k^{(\frac k2,0)}\D_{k-1}^{(\frac{k}{2},0)}}{\D_k^{(\frac{k}{2},1)}\D_{k-1}^{(\frac k2 -1,1)}}, &\text{ if k is even}. \label{eq:detinverseheven}
\end{align}
\end{subequations}
\end{theorem}

\subsection{The second inverse problem}
To proceed to the next step we recall that the original peakon problem \eqref{eq:m2CH_ode} was formulated
in the $x$ space.  To go back to the $x$ space we use the relation $h_j=m_je^{x_j}$ (see equation \eqref{eq:xLax}) to
arrive at the inverse formulae
relating the spectral data and the positions of peakons given by $x_j$.
\begin{theorem} \label{thm:inversex}
Given positive constants $m_j,n_j$ with distinct products $m_jn_j$, let $\Phi$ be the solution to the boundary value
problem \ref{eq:xLaxBVP} with associated spectral data $\{d\mu, c\}$.
Then the positions $x_j$ (of peakons) in the discrete measures
$m=2\sum_{j=1}^N m_j \delta_{x_j} $ and $n=2\sum_{j=1}^N n_j \delta_{x_j} $ can be expressed in
terms of the spectral data as:

\begin{subequations}
\begin{align}
&&x_{k'}&=\ln \frac
{\mathbf{e}_{[1,k-1]}\D_k^{(\frac{k+1}{2},0)}\D_{k-1}^{(\frac{k-1}{2},0)}}{m_{k'}\D_k^{(\frac{k-1}{2},1)}\D_{k-1}^{(\frac{k-1}{2},1)}},   &\text{ if k  is odd}, \label{eq:detinversexodd}\\
&&x_{k'}&= \ln \frac
{\mathbf{e}_{[1,k-1]}\D_k^{(\frac k2,0)}\D_{k-1}^{(\frac{k}{2},0)}}{m_{k'}\D_k^{(\frac{k}{2},1)}\D_{k-1}^{(\frac k2 -1,1)}}, &\text{ if k is even}, \label{eq:detinversexeven}
\end{align}
\end{subequations}
with $\D_k^{(l,p)}$ defined in \eqref{eq:calD}, $k'=N-k+1, \, 1\leq k\leq N$
and the convention that $\D_0^{l,p}=1$.
\end{theorem}

Finally, we can relax the condition that the products of masses $m_j,n_j$ be distinct.  Indeed, it suffices to observe that the Vandermonde determinants $\Delta_{[1, r]}(\mathbf{e}), \, r=k, k-1$,
cancel out in all expressions of the type
\begin{equation*}
\frac{\D_k^{( l_1,p_1)}\D_{k-1}^{(l_2,p_2)}}{\D_k^{(l_3,p_3)}\D_{k-1} ^{(l_4,p_4)}},
\end{equation*}
as a result of Theorem \ref{thm:detCSV} (see \eqref{eq:detCSV1} and \eqref{eq:detCSV2}).

In summary, we completed the full circle of starting with initial positions
of peakons, mapping them to the spectral data $\{d\mu, c\}$, while
in the last chapter we solved explicitly the inverse problem of
mapping back the spectral data  to the positions $x_j$ of peakons.
In all this the time was fixed.   In the following sections, we concentrate on the time evolution of the multipeakons of the 2-mCH.  In light of the difference of the value of $c$ in the Stieltjes transform according to whether $N$ is even or odd, the discussions will be presented separately for $N$ even, $N$ odd, respectively.

\section{Multipeakons for $N=2K$}\label{sec:evenpeakons}
For even $N$, $c=0$  (see Theorem \ref{thm:W}).  This impacts
the asymptotic behaviour of solutions.  Further
comments on differences between solutions for odd and even $N$
are in Section \ref{sec:redc_odd_even}.
\subsection{Closed formulae for $N=2K$}
If we assume that $x_1(0)<x_2(0)<\cdots<x_{2K}(0)$ then by continuity, this condition will hold at least in a small interval containing $t=0$.
At $t=0$ we solve the forward problem (see Section \ref{sec:FSM}) and obtain the
Weyl function $W(z)$ (see Theorem \ref{thm:W}) hence
the spectral measure $d\mu(0)=\sum_{j=1}^{K} b_j(0) \delta_{\zeta_j}$ supported on the set of
ordered eigenvalues $0<\zeta_1<\cdots<\zeta_K$.

With the help of Theorem \ref{thm:inversex} we obtain the following result.

\begin{theorem}\label{thm:peakon_even}
The 2-mCH \eqref{eq:m2CH} with the regularization
of the singular term $Q m$ given by $\avg{Q}m$ admits the multipeakon solution
\begin{equation}\label{eq:umultipeakoneven}
u(x,t)=\sum_{k=1}^{2K}m_{k'}\exp(-|x-x_{k'}(t)|), \qquad v(x,t)=\sum_{k=1}^{2K}n_{k'}\exp(-|x-x_{k'}(t)|),
\end{equation}
where $m_{k'}, n_{k'}$ are arbitrary positive constants, while $x_{k'}(t)$ are given by equations \eqref{eq:detinversexodd} and \eqref{eq:detinversexeven} corresponding to the peakon spectral measure
\begin{equation}\label{eq:peakon sm}
d\mu(t)=\sum_{j=1}^{K} b_j(t) \delta_{\zeta_j},
\end{equation}
with $b_j(t)=b_j(0)e^{\frac{2t}{\zeta_j}}, \, 0<b_j(0)$, and $c=0$ in \eqref{eq:calD}.
\end{theorem}
\begin{proof}
The only outstanding issue is the time evolution of $b_j$ or, more generally, the time evolution of the spectral measure.  Recall the Weyl function $W(z)$ defined in \eqref{eq:defWeyl}. By employing the time evolution \eqref{eq:tderqp},
one easily obtains
\begin{equation*}
\dot W=\frac 2z W-\frac{2u_+}{z},
\end{equation*}
which, in turn, implies
$\dot b_j=\frac{2}{\zeta_j} b_j, 1\leq j\leq K$, by use of
Corollary \ref{cor:spectrum}.

\end{proof}

In the following, we provide examples of multipeakons in the case of even $N$.  Before this is done, it is useful to examine the explicit formulas for the
CSV determinants following  Theorem \ref{thm:detCSV} (see equation \eqref{eq:calD} for notation). We remind the reader that the
eigenvalues $\zeta_j$ are positive and ordered $0<\zeta_1<\dots< \zeta_K$.
\begin{theorem} \label{thm:Dklm-peakon}  Let $N=2K, \, 0\leq l\leq K, \, 0\leq p, \,  p+l-1\leq k-l, \,1\leq k\leq 2K$ and let the peakon spectral measure be given by \eqref{eq:peakon sm}.
Then
\begin{enumerate}
\item
\begin{equation}\label{eq:Dklp}
\D_k^{(l,p)}=\abs{\Delta_{[1,k]}(\mathbf{e})}\sum_{I\in\binom{[1,K]}{l} }
\frac{\Delta^2_I(\mathbf{\zeta})\mathbf{b}_{I} \mathbf{\zeta}^p_I}{\Gamma_{[1,k],I}(\mathbf{e};\mathbf{\zeta})};
\end{equation}
\item in the asymptotic region $t\rightarrow + \infty$
\begin{equation} \label{eq:Dklp+}
\D_k^{(l,p)}=\abs{\Delta_{[1,k]}(\mathbf{e})}
\frac{\Delta^2_{[1,l]} (\mathbf{\zeta})\mathbf{b}_{[1,l]} \mathbf{\zeta}^p_{[1,l]}}{\Gamma_{[1,k],[1,l]}(\mathbf{e};\mathbf{\zeta})}\Big[1+\mathcal{O}(e^{-\alpha t})\Big], \quad \textrm{ for some } \alpha>0;
\end{equation}
\item in the asymptotic region $t \rightarrow -\infty$
\begin{equation}\label{eq:Dklp-}
\D_k^{(l,p)}=\abs{\Delta_{[1,k]}(\mathbf{e})}
\frac{\Delta^2_{[1,l]^*} (\mathbf{\zeta})\mathbf{b}_{[1,l]^*} \mathbf{\zeta}^p_{[1,l]^*}}{\Gamma_{[1,k],[1,l]^*}(\mathbf{e};\mathbf{\zeta})}\Big[1+\mathcal{O}(e^{\beta t})\Big]
, \quad 0<\beta,
\end{equation}
where $[1,l]^*=[l^*=K-l+1,1^*=K]$ $(\textrm{reflection of the interval }[1,K])$.
\end{enumerate}
\end{theorem}
Now we are ready to present examples of expressions
for positions $x_1, \cdots, x_{2K}$ of multipeakons based on formulas \eqref{eq:detinversexodd},
\eqref{eq:detinversexeven}, using \eqref{eq:Dklp} and 
$e_j=\frac{1}{m_{j'}n_{j'}}, \, j'=2K-j+1$.
\begin{example}[2-peakon solution; K=1]\label{ex:2peakon}
  \begin{align*}
     &x_1=\ln\left(\frac{b_1}{\zeta_1m_1(1+\zeta_1m_2n_2)}\right),\ \ \
     &x_2=\ln\left(\frac{b_1n_2}{1+\zeta_1m_2n_2}\right).
       \end{align*}
\end{example}

\begin{example}[4-peakon solution; K=2]\label{ex:4peakon}
\tiny{
  \begin{align*}
     &x_1=&\\
     &\ln\left(\frac{1}{m_1}\cdot\frac{b_1b_2(\zeta_2-\zeta_1)^2}{\zeta_1\zeta_2\left(b_1\zeta_1(1+\zeta_2m_2n_2)(1+\zeta_2m_3n_3)(1+\zeta_2m_4n_4)+b_2\zeta_2(1+\zeta_1m_2n_2)(1+\zeta_1m_3n_3)(1+\zeta_1m_4n_4)\right)}\right),&\\
     &x_2=\ln\left(n_2\cdot\frac{b_1b_2(\zeta_2-\zeta_1)^2\left(b_1(1+\zeta_2m_3n_3)(1+\zeta_2m_4n_4)+b_2(1+\zeta_1m_3n_3)(1+\zeta_1m_4n_4)\right)}{\left(b_1\zeta_1(1+\zeta_2m_3n_3)(1+\zeta_2m_4n_4)+b_2\zeta_2(1+\zeta_1m_3n_3)(1+\zeta_1m_4n_4)\right)}\right.&\\
     &\qquad\qquad \cdot\left.\frac{1}{\left(b_1\zeta_1(1+\zeta_2m_2n_2)(1+\zeta_2m_3n_3)(1+\zeta_2m_4n_4)+b_2\zeta_2(1+\zeta_1m_2n_2)(1+\zeta_1m_3n_3)(1+\zeta_1m_4n_4)\right)}\right),&\\
     &x_3=\\
     &\ln\left(\frac{1}{m_3}\cdot\frac{\left(b_1(1+\zeta_2m_4n_4)+b_2(1+\zeta_1m_4n_4)\right)\left(b_1(1+\zeta_2m_3n_3)(1+\zeta_2m_4n_4)+b_2(1+\zeta_1m_3n_3)(1+\zeta_1m_4n_4)\right)}{(1+\zeta_1m_4n_4)(1+\zeta_2m_4n_4)\left(b_1\zeta_1(1+\zeta_2m_3n_3)(1+\zeta_2m_4n_4)+b_2\zeta_2(1+\zeta_1m_3n_3)(1+\zeta_1m_4n_4)\right)}\right),&\\
     &x_4=\ln\left(n_4\cdot\frac{b_1(1+\zeta_2m_4n_4)+b_2(1+\zeta_1m_4n_4)}{(1+\zeta_1m_4n_4)(1+\zeta_2m_4n_4)}\right).&
       \end{align*}
       }
       \end{example}

\subsection{Global existence for $N=2K$}
Recall that our solution in Theorem \ref{thm:peakon_even} was obtained under the assumption that $x_1<x_2<\cdots<x_{2K}$. However, even if we start with the initial positions satisfying $x_1(0)<x_2(0)<\cdots<x_{2K}(0)$, the order might  cease to hold for sufficiently large times, in other words
some of the peakons might collide.  It is an interesting
question to understand the nature of collisions  but we leave this topic for future work.  In this subsection we give sufficient conditions in terms of the
spectrum and constant masses $m_j, n_j$ which ensure that
no collisions occur and thus the peakon solutions are global in $t$. The readers might want to consult Appendix \ref{app:globaleven} for a detailed proof.  Granted global existence, one can talk sensibly about the
large time behaviour of peakons.
\begin{theorem}\label{thm:global}
  Given arbitrary spectral data $$\{b_j>0, \, 0<\zeta_1<\zeta_2<\cdots< \zeta_K: 1\leq j \leq K \}, $$ suppose the masses $m_k,n_k$ satisfy
\begin{subequations}
\begin{align}
&\frac{\zeta_K^{\frac{k-1}{2}}}{\zeta_1^{\frac{k+1}{2}}}< m_{(k+1)' }n_{k'}, \qquad\qquad\qquad  \text { for all odd }k, \qquad   1\leq k\leq 2K-1,\\
&\frac{m_{(k+1)'}n_{(k+2)'}}{(1+m_{(k+1)'}n_{(k+1)'}\zeta_1)(1+m_{(k+2)'}n_{(k+2)'}\zeta_1)}<
\frac{\zeta_1^{\frac{k+1}{2}}}{\zeta_K^{\frac{k-1}{2}}}
\frac{2\, \textnormal{min}_j(\zeta_{j+1}-\zeta_j)^{k-1}}{(k+1)(\zeta_K-\zeta_1)^{k+1}}, \nonumber\\
&\qquad\qquad\qquad\qquad\qquad\qquad\qquad\text{ for all odd } k, \qquad 1\leq k\leq 2K-3. \label{eq:seccond}
\end{align}
\end{subequations}
Then the positions obtained from inverse formulas \eqref{eq:detinversexodd},
\eqref{eq:detinversexeven} are ordered $x_1<x_2<\cdots<x_{2K}$ and the multipeakon solutions \eqref{eq:umultipeakoneven} exist for arbitrary $t\in R$.
\end{theorem}

\subsection{Large time peakon asymptotics for $n=2K$}
Once the global existence of solutions is guaranteed, for example
by imposing sufficient conditions of Theorem \ref{thm:global}, one can study the asymptotic behaviour of multipeakon
solutions for large (positive and negative) time by employing Theorem \ref{thm:peakon_even} and  \ref{thm:Dklm-peakon}. More precisely, by using the formulae
for positions \eqref{eq:detinversexodd}, \eqref{eq:detinversexeven}, as
well as asymptotic evaluations of determinants  \eqref{eq:Dklp+} and
\eqref{eq:Dklp-}, one arrives at

\begin{theorem}\label{thm:evenass} Suppose the masses $m_j,n_j$ satisfy the conditions of Theorem \ref{thm:global}.  Then the asymptotic position of a $k$-th (counting from the right) peakon as $t\rightarrow+\infty$ is given by
\begin{subequations}
\begin{align*}
&x_{k'}=\frac{2t}{\zeta_{\frac{k+1}{2}}}+
\ln\frac{b_{\frac{k+1}{2}}(0)\mathbf{e}_{[1,k-1]} \Delta^2_{[1,\frac{k-1}{2}],\{\frac{k+1}{2}\}}(\mathbf{\zeta})}{m_{k'} \Gamma_{[1,k], \{\frac{k+1}{2}\}}(\mathbf{e}; \mathbf{\zeta}) \mathbf{\zeta}^2_{[1,\frac{k-1}{2}]}}+\mathcal{O}(
e^{-\alpha_k t}), \nonumber\\
&\qquad\qquad\qquad\qquad\qquad\qquad\quad\qquad\quad\qquad\textrm{ for some positive } \alpha_k\, \textrm{ and odd } k, \\
&x_{k'}=\frac{2t}{\zeta_{\frac{k}{2}}}+
\ln\frac{b_{\frac{k}{2}}(0)\mathbf{e}_{[1,k-1]} \Delta^2_{[1,\frac{k}{2}-1],\{\frac{k}{2}\}}(\mathbf{\zeta})}{m_{k'} \Gamma_{[1,k-1], \{\frac{k}{2}\}}(\mathbf{e}; \mathbf{\zeta}) \mathbf{\zeta}^2_{[1,\frac{k}{2}-1]}\zeta_{\frac k2}}+\mathcal{O}(
e^{-\alpha_k t}), \nonumber\\
&\qquad\qquad\qquad\qquad\qquad\qquad\quad\qquad\qquad\quad\textrm{ for some positive } \alpha_k\, \textrm{ and even } k,\\
&x_{k'}-x_{(k+1)'}=\ln m_{(k+1)'}n_{k'} \zeta_{\frac{k+1}{2}}
+\mathcal{O}(e^{-\alpha_k t}), \\
&\qquad\qquad\qquad\qquad\qquad\qquad\quad\qquad\qquad\quad\textrm{ for some positive } \alpha_k\, \textrm{ and odd } k.
\end{align*}
\end{subequations}

  Likewise, as $t\rightarrow-\infty$, using the notation of Theorem \ref{thm:Dklm-peakon}, the asymptotic position of the $k$-th peakon is given by
  \begin{subequations}
  \begin{align*}
&x_{k'}=\frac{2t}{\zeta_{(\frac{k+1}{2})^*}}+
\ln\frac{b_{(\frac{k+1}{2})^*}(0)\mathbf{e}_{[1,k-1]} \Delta^2_{([1,\frac{k-1}{2}])^*,\{(\frac{k+1}{2})^*\}}(\mathbf{\zeta})}{m_{k'} \Gamma_{[1,k], \{(\frac{k+1}{2})^*\}}(\mathbf{e}; \mathbf{\zeta}) \mathbf{\zeta}^2_{[1,\frac{k-1}{2}]^*}}+\mathcal{O}(
e^{\beta_k t}), \\
&\qquad\qquad\qquad\qquad\qquad\qquad\quad\qquad\quad\qquad \textrm{ for some positive } \beta_k\, \textrm{ and odd } k, \\
&x_{k'}=\frac{2t}{\zeta_{(\frac{k}{2})^*}}+
\ln\frac{b_{(\frac{k}{2})^*}(0)\mathbf{e}_{[1,k-1]} \Delta^2_{[1,\frac{k}{2}-1]^*,\{(\frac{k}{2})^*\}}(\mathbf{\zeta})}{m_{k'} \Gamma_{[1,k-1], \{(\frac{k}{2})^*\}}(\mathbf{e}; \mathbf{\zeta}) \mathbf{\zeta}^2_{[1,\frac{k}{2}-1]^*}\zeta_{(\frac k2)^*}}+\mathcal{O}(
e^{\beta_k t}), \\
&\qquad\qquad\qquad\qquad\qquad\qquad\quad\qquad\quad\qquad \textrm{ for some positive } \beta_k\, \textrm{ and even } k,\\
&x_{k'}-x_{(k+1)'}=\ln m_{(k+1)'}n_{k'} \zeta_{(\frac{k+1}{2})^*}
+\mathcal{O}(e^{\beta_k t}), \\
&\qquad\qquad\qquad\qquad\qquad\qquad\quad\qquad\quad\qquad \textrm{ for some positive } \beta_k\, \textrm{ and odd } k.
\end{align*}
\end{subequations}
\end{theorem}

\begin{remark}
It follows from the above theorem that multipeakons
of the 2-mCH equation exhibit \textit{Toda-like sorting properties} of asymptotic speeds and \textit{ an asymptotic pairing}.  The latter can be partially explained by  the fact that there are  $K$ available eigenvalues to match $ 2K$ asymptotic speeds. Similar features were also observed in the mCH equation \cite{chang-szmigielski-m1CHlong}, as well as the interlacing cases of the 2-mCH equation \cite{chang2016multipeakons}. It is clear now that these two features extend to the non-interlacing cases as well.  \end{remark}

We end this section by providing graphs of a concrete 4-peakon solution.
 Let $K=2$, and $ b_1(0)=10,\ b_2(0)=1,\ \zeta_1=0.3,\ \zeta_2=3,\ m_1=8,\ m_2=24,\ m_3=5,\ m_4=10,\ n_1=12,\ n_2=10,\ n_3=24,\ n_4=16$. It is easy to check that the condition in Theorem \ref{thm:global} is satisfied. Hence,  the order of $\{x_k, k = 1, 2, 3, 4\}$ will be preserved and one can use the explicit formulae for the 4-peakon solution, resulting in the following sequence of graphs (Figure \ref{fig_4peakon}), illustrating the asymptotic pairing of peakons.
\begin{figure}[h!]
  \centering
  \resizebox{1.1\textwidth}{!}{
  \includegraphics{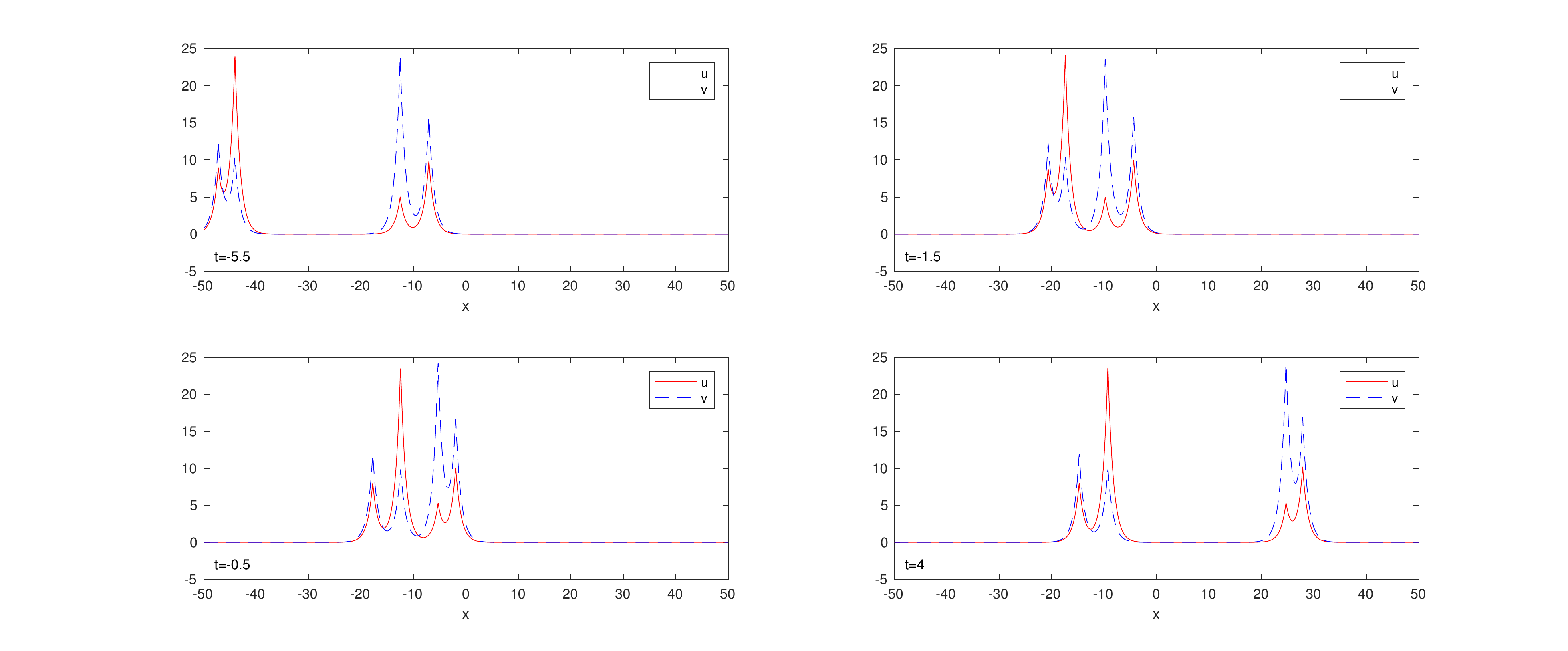}
  }
   \caption{Graphs of 4-peakon at times $t=-5.5,\ -1.5,\ -0.5, \ 4$ in the case of $ b_1(0)=10,\ b_2(0)=1,\ \zeta_1=0.3,\ \zeta_2=3,\ m_1=8,\ m_2=24,\ m_3=5,\ m_4=10,\ n_1=12,\ n_2=10,\ n_3=24,\ n_4=16.$}
   \label{fig_4peakon}
\end{figure}

\section{Multipeakons for $N=2K+1$}\label{sec:oddpeakons}
This section is devoted to the corresponding result for $N=2K+1$, which is presented in a way parallel to the previous section on the even case.
The main source of difference between the two cases  is of course the presence of
the positive shift $c$ which impacts the evaluations of
the CSV determinants as illustrated by Theorem \ref{thm:detCSV}, in particular
formula \eqref{eq:detCSV2}.   Nevertheless, we will present a comparison in the form of a correspondence between the odd case and the even case  in Section \ref{sec:redc_odd_even}.
\subsection{Closed formulae for $N=2K+1$}

As before we assume that $x_1(0)<x_2(0)<\cdots<x_{2K+1}(0)$.  Then this condition will hold at least in a small interval containing $t=0$.  The following \textit{local existence} result follows from Theorem \ref{thm:inversex}.

\vspace{1cm}

\begin{theorem}\label{thm:peakon_odd}
The 2-mCH equation \eqref{eq:m2CH} with the regularization
of the singular term $Q m$ given by $\avg{Q}m$ admits the multipeakon solution
\begin{equation}\label{eq:umultipeakoneven}
u(x,t)=\sum_{k=1}^{2K+1}m_{k'} \exp(-|x-x_{k'}(t)|),\qquad v(x,t)=\sum_{k=1}^{2K+1}n_{k'}\exp(-|x-x_{k'}(t)|),
\end{equation}
where $m_{k'}, n_{k'}$ are arbitrary positive constants, while $x_{k'}(t)$ are given by equations \eqref{eq:detinversexodd} and \eqref{eq:detinversexeven} corresponding to the peakon spectral measure
\begin{equation}\label{eq:peakon sm(odd)}
d\mu=\sum_{j=1}^{K} b_j(t) \delta_{\zeta_j},
\end{equation}
$b_j(t)=b_j(0)e^{\frac{2t}{\zeta_j}}, \, 0<b_j(0)$, with ordered eigenvalues $0<\zeta_1<\cdots<\zeta_K$ and $c(t)=c(0)>0$ in \eqref{eq:calD}.
\end{theorem}
\begin{proof}
Similar to the even case,   it is clear that the Weyl function $W(z)$ is defined in \eqref{eq:defWeyl},
undergoes the time evolution obtained earlier in the proof of Theorem \ref{thm:peakon_even}, namely,
\begin{equation*}
\dot W=\frac 2z W-\frac{2u_+}{z},
\end{equation*}
which, in turn, implies
$\dot b_j=\frac{2}{\zeta_j} b_j, 1\leq j\leq K$ as well as $\dot c=0$ by virtue of
Corollary \ref{cor:spectrum}.
 The rest of the argument is the same as for the even case.
\end{proof}

By using the above theorem, it is not hard to work out two simplest examples of solutions.  Before we do that, however,
we will examine the evaluation of
CSV determinants presented in Theorem \ref{thm:detCSV} (see equation \eqref{eq:calD} for notation), with due care to two facts: $N=2K+1$ and $c>0$.
The proof follows from the same steps as in Theorem \ref{thm:Dklm-peakon} and
we omit it.

\begin{theorem} \label{thm:Dklp-peakon(odd)}  Let $N=2K+1, \, 1\leq k\leq 2K+1, \, 0\leq l\leq K+1, \, 0\leq p, \,  p+l-1\leq k-l, $ and let the peakon spectral measure be given by \eqref{eq:peakon sm(odd)} and a shift $c>0$.
Then
\begin{enumerate}
\item
\begin{subequations}
\begin{align}
&\D_k^{(l,p)}=\abs{\Delta_{[1,k]}(\mathbf{e})}\sum_{I\in\binom{[1,K]}{l} }
\frac{\Delta^2_I(\mathbf{\zeta})\mathbf{b}_{I} \mathbf{\zeta}^p_I}{\Gamma_{[1,k],I}(\mathbf{e};\mathbf{\zeta})}, \nonumber \\
&\qquad\qquad\qquad\qquad\qquad\quad \text{ if} \quad p+l-1<k-l, \quad k\leq 2K+1; \label{eq:Dklpodd1}\\
&\D_k^{(l,p)}=\abs{\Delta_{[1,k]}(\mathbf{e})}\Big(\sum_{I\in\binom{[1,K]}{l} }
\frac{\Delta^2_I(\mathbf{\zeta})\mathbf{b}_{I} \mathbf{\zeta}^p_I}{\Gamma_{[1,k],I}(\mathbf{e};\mathbf{\zeta})} +c\sum_{I\in\binom{[1,K]}{l-1} }
\frac{\Delta^2_I(\mathbf{\zeta})\mathbf{b}_{I} \mathbf{\zeta}^p_I}{\Gamma_{[1,k],I}(\mathbf{e};\mathbf{\zeta})}\Big)
\nonumber\\
&\qquad\qquad\qquad\qquad\qquad\quad  \text{ if} \quad p+l-1=k-l, \quad k\leq 2K+1; \label{eq:Dklpodd2}
\end{align}
\end{subequations}
with the proviso that the first term inside the bracket is set to zero if $l=K+1$, which only happens when
$k=2K+1, p=0$.
\item In the asymptotic region $t\rightarrow + \infty$
\begin{subequations}
\begin{align}
 &\D_k^{(l,p)}=\abs{\Delta_{[1,k]}(\mathbf{e})}
\frac{\Delta^2_{[1,l]} (\mathbf{\zeta})\mathbf{b}_{[1,l]} \mathbf{\zeta}^p_{[1,l]}}{\Gamma_{[1,k],[1,l]}(\mathbf{e};\mathbf{\zeta})}\Big[1+\mathcal{O}(e^{-\alpha t})\Big], \quad 0< \alpha, \nonumber\\
&\qquad\qquad\qquad\qquad\qquad\qquad \text{ if } \quad 0\leq l\leq K; \label{eq:Dklp+odd}\\
 &\D_{2K+1}^{(K+1,0)}=c\abs{\Delta_{[1,2K+1]}(\mathbf{e})}
\frac{\Delta^2_{[1,K]} (\mathbf{\zeta})\mathbf{b}_{[1,K]} }{\Gamma_{[1,2K+1],[1,K]}(\mathbf{e};\mathbf{\zeta})}, \nonumber\\
&\qquad\qquad\qquad\qquad\qquad\qquad \text{ if } \quad k=2K+1, l=K+1, p=0.  \label{eq:Dklp+odd-c}
\end{align}
\end{subequations}
\item In the asymptotic region $t \rightarrow -\infty$
\begin{subequations}
\begin{align}
&\D_k^{(l,p)}=\abs{\Delta_{[1,k]}(\mathbf{e})}
\frac{\Delta^2_{[1,l]^*} (\mathbf{\zeta})\mathbf{b}_{[1,l]^*} \mathbf{\zeta}^p_{[1,l]^*}}{\Gamma_{[1,k],[1,l]^*}(\mathbf{e};\mathbf{\zeta})}\Big[1+\mathcal{O}(e^{\beta t})\Big]
, \quad 0<\beta,\notag\\
&\qquad\qquad\qquad\qquad\qquad\quad \text{ if} \quad p+l-1<k-l, \quad k\leq 2K+1; \label{eq:Dklp-odd1}\\
&\D_k^{(l,p)}=c\abs{\Delta_{[1,k]}(\mathbf{e})}
\frac{\Delta^2_{[1,l-1]^*} (\mathbf{\zeta})\mathbf{b}_{[1,l-1]^*} \mathbf{\zeta}^p_{[1,l-1]^*}}{\Gamma_{[1,k],[1,l-1]^*}(\mathbf{e};\mathbf{\zeta})}\Big[1+\mathcal{O}(e^{\beta t})\Big]
, \quad 0<\beta,  \notag\\
&\qquad\qquad\qquad\qquad\qquad\quad \text{ if}  \quad p+l-1=k-l,\quad  k<2K+1;\label{eq:Dklp-odd2}\\
&\D_{2K+1}^{(K+1,0)}=c\abs{\Delta_{[1,2K+1]}(\mathbf{e})}
\frac{\Delta^2_{[1,K]} (\mathbf{\zeta})\mathbf{b}_{[1,K]} }{\Gamma_{[1,2K+1],[1,K]}(\mathbf{e};\mathbf{\zeta})}, \notag\\
&\qquad\qquad\qquad\qquad\qquad\quad \text{ if } \quad k=2K+1, l=K+1, p=0, \label{eq:Dklp-odd3}
\end{align}
\end{subequations}
where, as before, $[1,l]^*=[l^*=K-l+1,1^*=K]$.
\end{enumerate}
\end{theorem}

There exists a relation between formulae with $c>0$ and $c=0$. Indeed, by comparing formulas \eqref{eq:Dklpodd1}
 and \eqref{eq:Dklpodd2} with \eqref{eq:Dklp},
we arrive at the detailed dependence on $c$.
 \begin{corollary} \label{cor:Dklp-c-dependence}
 Let $N=2K+1,\,  1\leq k\leq 2K+1, \,\, 0\leq l\leq K+1, \, 0\leq p, \,  p+l-1\leq k-l$,  and let the peakon spectral measure be given by \eqref{eq:peakon sm(odd)} and a shift $c>0$.
Then
\begin{subequations}
\begin{align}
&\D_k^{(l,p)}(c)=\D_{k}^{(l,p)}(0), \  \text{ if} \  p+l-1<k-l, \  k\leq 2K+1; \label{eq:Dklpodd1-c-dep}\\
&\D_k^{(l,p)}(c)=\D_{k}^{(l,p)}(0)+c\D_k^{(l-1,p)}(0),
\ \text{ if} \ p+l-1=k-l, \  k\leq 2K+1; \label{eq:Dklpodd2-c-dep}
\end{align}
\end{subequations}
with the convention that the first term in \eqref{eq:Dklpodd2-c-dep} is set to zero if $l=K+1, k=2K+1,p=0$.
\end{corollary}
Below the reader will find two examples of
explicit peakon solutions for odd $N=2K+1$.

\begin{example}[1-peakon solution; $K=0$ ] \label{ex:1peakon}
  \[
  x_1=\ln\left(\frac{c}{m_1}\right).
  \]
\end{example}

\begin{example}[3-peakon solution; $K=1$]\label{ex:3peakon}
  \begin{align*}
     &x_1=\ln\left(\frac{b_1c}{\zeta_1m_1\left(b_1\zeta_1m_2n_2m_3n_3+c(1+\zeta_1m_2n_2)(1+\zeta_1m_3n_3)\right)}\right),&\\
     &x_2=\ln\left(\frac{b_1n_2}{b_1\zeta_1m_2n_2m_3n_3+c(1+\zeta_1m_2n_2)(1+\zeta_1m_3n_3)}\left(\frac{b_1m_3n_3}{1+\zeta_1m_3n_3}+c\right)\right),&\\
     &x_3=\ln\left(\frac{1}{m_3}\left(\frac{b_1m_3n_3}{1+\zeta_1m_3n_3}+c\right)\right).
  \end{align*}

\end{example}

\subsection{Global existence for $N=2K+1$}
Similar to the even case, we can also provide a sufficient condition to ensure the global existence
 of peakon solutions when $N=2K+1$.  The main result is stated below while its proof is relegated to  Appendix \ref{app:globalodd}.
\begin{theorem}\label{thm:global(odd)}
  Given arbitrary spectral data $$\{b_j>0, \, 0<\zeta_1<\zeta_2<\cdots< \zeta_K,\,
  c>0: 1\leq j \leq K \}, $$ suppose the masses $m_k,n_k$ satisfy
\begin{subequations}
\begin{align*}
&\frac{1}{m_{(k+1)' }n_{k'}}<\frac{\zeta_1^{\frac{k+1}{2}}}{\zeta_K^{\frac{k-1}{2}}}\textrm{min}\{1, \hat \beta\},  &\text { for all odd }k, \qquad   1\leq k\leq 2K-1, \\
&\frac{1}{m_{(k+2)'}n_{(k+1)'}}<\frac{\zeta_1^{\frac{k+1}{2}}}{\zeta_K^{\frac{k-1}{2}}}\textrm{min}\{1, \hat \beta_1\}, &\text{ for all odd } k, \qquad 1\leq k\leq 2K-1,
\end{align*}
\end{subequations}
where
\begin{align*}
\hat \beta&=\begin{cases} \frac{2\zeta_K\, \textrm{min}_j(\zeta_{j+1}-\zeta_j)^{k-3}}{\zeta_1(k-1)(\zeta_K-\zeta_1)^{k-1}}\frac{(1+m_{k'}n_{k'}\zeta_1)(1+m_{(k+1)'}n_{(k+1)'}\zeta_1)}{m_{k'}n_{k'} m_{(k+1)'}n_{(k+1)'}},   &\substack{\text{ for all odd } k,\\ 3\leq k\leq 2K-1,}  \\
+\infty,  &\text{ for } k=1, \end{cases}\\\\
\hat \beta_1&=\frac{2\, \textrm{min}_j(\zeta_{j+1}-\zeta_j)^{k-1}}{(k+1)(\zeta_K-\zeta_1)^{k+1}}\frac{(1+m_{(k+1)'}n_{(k+1)'}\zeta_1)(1+m_{(k+2)'}n_{(k+2)'}\zeta_1)}{m_{(k+1)'}n_{(k+1)'} m_{(k+2)'}n_{(k+2)'}}.
\end{align*}
Then the positions obtained from inverse formulas \eqref{eq:detinversexodd},
\eqref{eq:detinversexeven} are ordered $x_1<x_2<\cdots<x_{2K+1}$ and the multipeakon solutions \eqref{eq:umultipeakoneven} exist for arbitrary $t\in \R$.
\end{theorem}

\subsection{Large time peakon asymptotics for $N=2K+1$}
Again, based on the global existence of multipeakons guaranteed by Theorem \ref{thm:global(odd)}, one can investigate the long time asymptotics
of  global multipeakon solutions.   After a straightforward, but tedious, computation using the formulas
for positions \eqref{eq:detinversexodd}, \eqref{eq:detinversexeven} (for $N=2K+1$), as
well as asymptotic evaluations of determinants  presented in Theorem \ref{thm:Dklp-peakon(odd)}, we obtain
\begin{theorem}\label{thm:oddass} Suppose the masses $m_j, n_j$ satisfy the conditions of Theorem \ref{thm:global(odd)}.  Then the asymptotic position of a $k$-th (counting from the right) peakon as $t\rightarrow+\infty$ is given by
\begin{subequations}
\begin{align*}
&x_{k'}=\frac{2t}{\zeta_{\frac{k+1}{2}}}+
\ln\frac{b_{\frac{k+1}{2}}(0)\mathbf{e}_{[1,k-1]} \Delta^2_{[1,\frac{k-1}{2}],\{\frac{k+1}{2}\}}(\mathbf{\zeta})}{m_{k'} \Gamma_{[1,k], \{\frac{k+1}{2}\}}(\mathbf{e}; \mathbf{\zeta}) \mathbf{\zeta}^2_{[1,\frac{k-1}{2}]}}+\mathcal{O}(
e^{-\alpha_k t}), \qquad  \alpha_k>0\\
& \hspace{9cm}\textrm{ and odd } k\leq 2K-1; \\
&x_{(2K+1)'}=\ln \frac{c \mathbf{e}_{[1,2K]}}{m_{(2K+1)'}\mathbf{\zeta}_{[1,K]}^2}
+\mathcal{O}(
e^{-\alpha t}), \qquad \alpha>0 ; \\
&x_{k'}=\frac{2t}{\zeta_{\frac{k}{2}}}+
\ln\frac{b_{\frac{k}{2}}(0)\mathbf{e}_{[1,k-1]} \Delta^2_{[1,\frac{k}{2}-1],\{\frac{k}{2}\}}(\mathbf{\zeta})}{m_{k'} \Gamma_{[1,k-1], \{\frac{k}{2}\}}(\mathbf{e}; \mathbf{\zeta}) \mathbf{\zeta}^2_{[1,\frac{k}{2}-1]}\zeta_{\frac k2}}+\mathcal{O}(
e^{-\alpha_k t}), \qquad \alpha_k>0 \\
& \hspace{9cm} \textrm{ and even } k\leq 2K;\\
&x_{k'}-x_{(k+1)'}=\ln m_{(k+1)'}n_{k'} \zeta_{\frac{k+1}{2}}
+\mathcal{O}(e^{-\alpha_k t}),  \qquad \alpha_k>0\\
& \hspace{9cm} \textrm{ and odd } k \leq 2K-1.
\end{align*}
\end{subequations}

  Likewise, as $t\rightarrow-\infty$, using the notation of Theorem \ref{thm:Dklm-peakon}, the asymptotic position of the $k$-th peakon is given by
  \begin{subequations}
  \begin{align*}
  &x_{k'}=\frac{2t}{\zeta_{(\frac{k-1}{2})^*}}+
\ln\frac{b_{(\frac{k-1}{2})^*}(0)\mathbf{e}_{[1,k-1]} \Delta^2_{[1,\frac{k-1}{2}-1]^*,\{(\frac{k-1}{2})^*\}}(\mathbf{\zeta})}{m_{k'} \Gamma_{[1,k-1], \{(\frac{k-1}{2})^*\}}(\mathbf{e}; \mathbf{\zeta}) \mathbf{\zeta}^2_{[1,\frac{k-1}{2}-1]^*}\zeta_{(\frac {k-1}{2})^*}}+\mathcal{O}(
e^{\beta_k t}), \\ &\hspace{7cm} \beta_k>0 \textrm{ and odd  }1<k\leq 2K+1;\\
&x_{1'}=\ln\frac{c}{m_{1'}}+\mathcal{O}(
e^{\beta_k t}), \qquad  \beta_k>0;\\
&x_{k'}=\frac{2t}{\zeta_{(\frac{k}{2})^*}}+
\ln\frac{b_{(\frac{k}{2})^*}(0)\mathbf{e}_{[1,k-1]} \Delta^2_{([1,\frac{k}{2}-1])^*,\{(\frac{k}{2})^*\}}(\mathbf{\zeta})}{m_{k'} \Gamma_{[1,k], \{(\frac{k}{2})^*\}}(\mathbf{e}; \mathbf{\zeta}) \mathbf{\zeta}^2_{[1,\frac{k}{2}-1]^*}}+\mathcal{O}(
e^{\beta_k t}), \\\
&
\hspace{8cm} \beta_k>0 \textrm{ and even } k\leq 2K; \\
&x_{k'}-x_{(k+1)'}=\ln m_{(k+1)'}n_{k'} \zeta_{(\frac{k}{2})^*}
+\mathcal{O}(e^{\beta_k t}), \\
& \hspace{8cm}\beta_k>0
 \textrm{ and even } k\leq 2K.
\end{align*}
\end{subequations}
\end{theorem}

\begin{remark}
Similar to the even case, the Toda-like sorting property can also be observed in this case.  It is perhaps interesting to examine the role of the constant $c$.  The constant $c$ is playing the role of an additional eigenvalue $\zeta_{K+1}=\infty$ resulting in the formal asymptotic speed $0$.
We observe that for large positive times the first particle counting from the left comes to a stop, while
the remaining $2K$ peakons form pairs of bound states akin to what is
occurring for even $N$, effectively sharing in pairs the remaining $K$ speeds.
By contrast, for large negative times, the first particle counting from the right slows to a halt,
while the remaining peakons form pairs.  Note that a similar phenomenon also occurs in the mCH.
\end{remark}

At the end of this section, we present graphs for a concrete 3-peakon solution, which confirm our theoretical predictions. Let $K=1$, and $ b_1(0)=1,\ c=3,\ \zeta_1=5,\ m_1=3,\ m_2=2,\ m_3=1.6,\ n_1=1.8,\ n_2=3,\ n_3=2.2.$ Then the sufficient conditions in Theorem \ref{thm:global(odd)} are satisfied. Hence the order of $\{x_k, k = 1, 2, 3\}$ will be preserved  and one can use the explicit formulae for the 3-peakon solution, resulting in the following sequence of graphs (Figure \ref{fig_3peakon}).
\begin{figure}[h!]
  \centering
  \resizebox{1.1\textwidth}{!}{
  \includegraphics{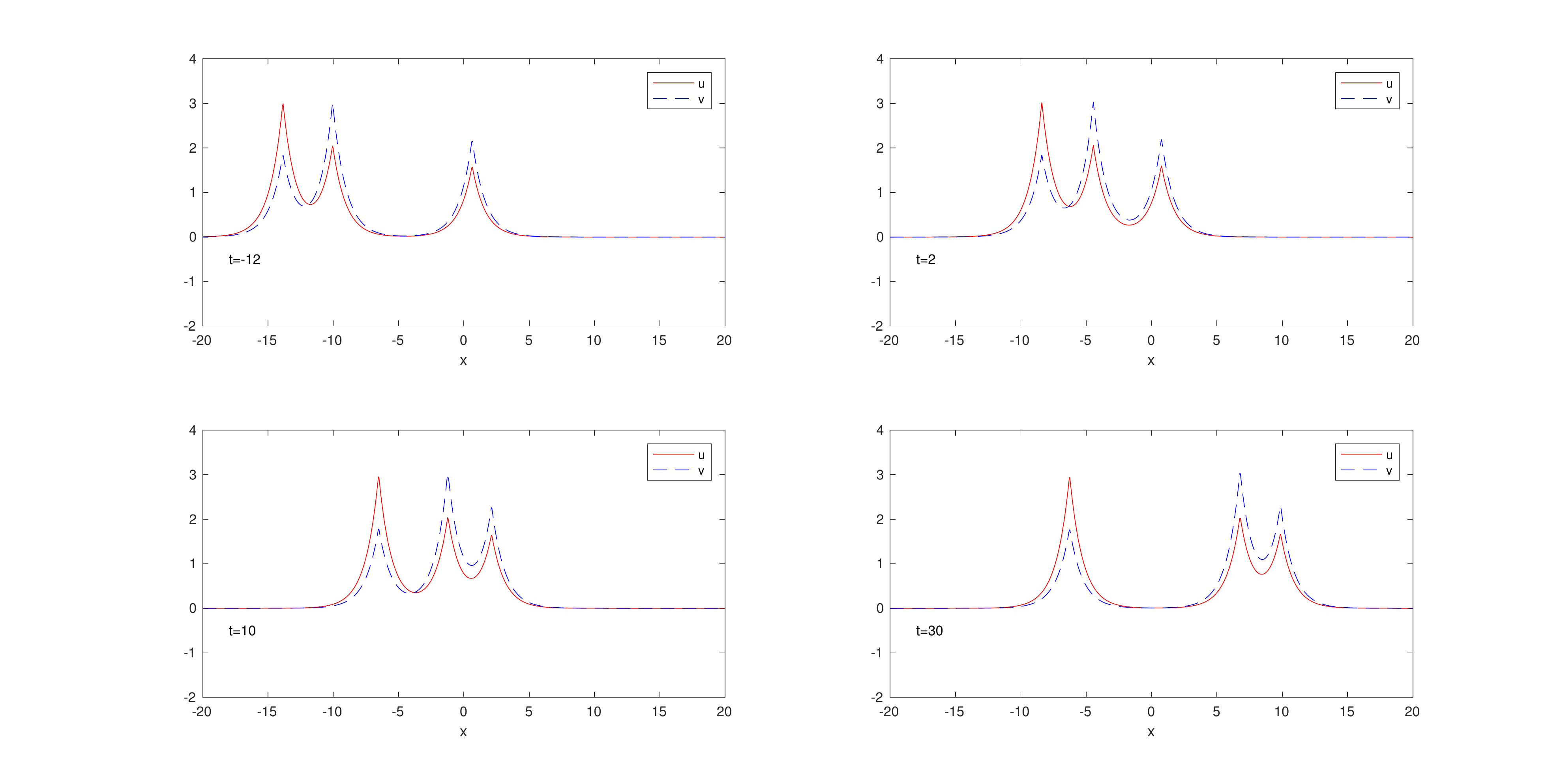}
  }
   \caption{Graphs of 3-peakon at time $t=-12,\ 2,\ 10, \ 30$ in the case of $ b_1(0)=1,\ c=3,\ \zeta_1=5,\ m_1=3,\ m_2=2,\ m_3=1.6,\ n_1=1.8,\ n_2=3,\ n_3=2.2.$}
   \label{fig_3peakon}
\end{figure}

\section{Reductions of multipeakons}
We recall that the 2-mCH \eqref{eq:m2CH} is a two-component integrable generalization of the mCH \eqref{eq:m1CH}. In the present contribution we constructed explicit formulae for generic multipeakons of the  2-mCH by using the inverse spectral method. Yet, in our past work we constructed the  so-called interlacing multipeakons of the 2-mCH \cite{chang2016multipeakons}, while in \cite{chang-szmigielski-m1CHlong} we gave explicit formulae for multipeakons of the mCH.  In this section, as if to close the circle, we would like to show how the formulae obtained in previous sections reduce to those special cases.
\subsection{From odd case to even case} \label{sec:redc_odd_even}
In Section \ref{sec:evenpeakons} and \ref{sec:oddpeakons}, we presented the multipeakon formulae according to the parity of the number $N$ of masses. In this subsection, we will show how the multipeakon formula in the odd case can be used to derive the mutipeakon formula for the even case. Consider the multipeakons
$$u=\sum_{j=1}^{2K+1} m_j (t)e^{-\abs{x-x_j(t)}},\qquad v=\sum_{j=1}^{2K+1} n_j (t)e^{-\abs{x-x_j(t)}}.$$
Suppose
$m_1\rightarrow0,\, n_1\rightarrow0$,
then
$h_1\rightarrow0,\, g_1\rightarrow0$,
since
$h_j=m_je^{x_j},\,  g_j=n_je^{-x_j}$.
Moreover, by continuity, it is not hard to see that
$c\rightarrow0$ in view of
\eqref{eq:c} and results from Section \ref{sec:FSM}.

Now, with the help of multipeakon formulae \eqref{eq:detinversexodd}-\eqref{eq:detinversexeven}, we also see that when $c\rightarrow0$, 
then
$\lim_{c\rightarrow0}m_1e^{-|x-x_1|}=0$.  Moreover, we can rewrite the multipeakon formulae \eqref{eq:detinversexodd}-\eqref{eq:detinversexeven} in the odd case as
\begin{subequations}
\begin{align}
&&x_{2K-2k+3}&=\ln \frac
{\mathbf{e}_{[1,2k-2]}\D_{2k-1}^{(k,0)}(c)\, \D_{2k-2}^{(k-1,0)}(c)}{m_{2K-2k+3}\D_{2k-1}^{(k-1,1)}(c)\,\D_{2k-2}^{(k-1,1)}(c)},   &&1\leq k\leq K+1, \\
&&x_{2K-2k+2}&= \ln \frac
{\mathbf{e}_{[1,2k-1]}\D_{2k}^{(k,0)}(c)\, \D_{2k-1}^{(k,0)}(c)}{m_{2K-2k+2}\D_{2k}^{(k,1)}(c)\,\D_{2k-1}^{(k -1,1)}(c)}, &&1\leq k\leq K,
\end{align}
\end{subequations}
where $e_j=\frac{1}{m_{2K+2-j}n_{2K+2-j}}$.
Note, however, that neither $m_1$ nor $n_1$ appear in the expressions for ~$x_2,x_3,\cdots, x_{2K+1}$. Thus, we have
\begin{subequations}
\begin{align}
&&\lim_{c\rightarrow0}x_{2K-2k+3}&=\ln \frac
{\mathbf{e}_{[1,2k-2]}\D_{2k-1}^{(k,0)}(0)\, \D_{2k-2}^{(k-1,0)}(0)}{m_{2K-2k+3}\D_{2k-1}^{(k-1,1)}(0)\,\D_{2k-2}^{(k-1,1)}(0)},   &&1\leq k\leq K, \\
&&\lim_{c\rightarrow0}x_{2K-2k+2}&= \ln \frac
{\mathbf{e}_{[1,2k-1]}\D_{2k}^{(k,0)}(0)\, \D_{2k-1}^{(k,0)}(0)}{m_{2K-2k+2}\D_{2k}^{(k,1)}(0)\,\D_{2k-1}^{(k -1,1)}(0)}, &&1\leq k\leq K.
\end{align}
\end{subequations}
With the notation
$$\tilde x_j =\lim_{c\rightarrow0}x_{j+1},\qquad \tilde m_j =m_{j+1},\qquad \tilde n_j =n_{j+1},$$
as well as
$$\tilde e_j=\frac{1}{\tilde m_{2K+1-j}\tilde n_{2K+1-j}},\qquad \mathbf{\tilde e}_{[1,j]}=\tilde e_1\cdots \tilde e_j,\qquad \tilde \D_k^{l,p}=\D_k^{l,p}$$
where $ \tilde \D_k^{l,p}$ is obtained by replacing $m_{j+1},n_{j+1}$ by $\tilde m_j, \tilde n_j$ in the corresponding CSV matrix,
we obtain
\begin{subequations}
\begin{align}
&&\tilde x_{2K-2k+2}&=\ln \frac
{\mathbf{\tilde e}_{[1,2k-2]}\tilde \D_{2k-1}^{(k,0)}(0)\, \tilde \D_{2k-2}^{(k-1,0)}(0)}{\tilde m_{2K-2k+2}\tilde \D_{2k-1}^{(k-1,1)}(0)\,\tilde \D_{2k-2}^{(k-1,1)}(0)},   &&1\leq k\leq K, \\
&&\tilde x_{2K-2k+1}&= \ln \frac
{\mathbf{\tilde e}_{[1,2k-1]}\tilde \D_{2k}^{(k,0)}(0)\, \tilde \D_{2k-1}^{(k,0)}(0)}{\tilde m_{2K-2k+1}\tilde \D_{2k}^{k,1)}(0)\,\tilde \D_{2k-1}^{(k -1,1)}(0)}, &&1\leq k\leq K.
\end{align}
\end{subequations}
Comparing the formulae \eqref{eq:detinversexodd}-\eqref{eq:detinversexeven} in the even case, these are nothing but the formulae for the positions $\tilde x_j$ corresponding the multipeakon ansatz:
$$u=\sum_{j=1}^{2K} \tilde m_j (t)e^{-\abs{x-\tilde x_j(t)}},\qquad v=\sum_{j=1}^{2K+1} \tilde n_j (t)e^{-\abs{x-\tilde x_j(t)}}.$$

So far, we have shown how the formulae for the odd case reduce to those for the even case. It is not hard to see that a comparison between Examples \ref{ex:2peakon} and \ref{ex:3peakon} supports our claim above.

\subsection{From the even case to the interlacing case}
Let us consider the multipeakon ansatz with the even number of masses:
$$u=\sum_{j=1}^{2K} m_j (t)e^{-\abs{x-x_j(t)}},\qquad v=\sum_{j=1}^{2K} n_j (t)e^{-\abs{x-x_j(t)}}. $$
Then the reduction
$m_{2j}\rightarrow0,\,  n_{2j-1}\rightarrow0$
gives the interlacing ansatz
$$u=\sum_{j=1}^{K} m_{2j-1} (t)e^{-\abs{x-x_{2j-1}(t)}},\qquad v=\sum_{j=1}^{K} n_{2j}(t)e^{-\abs{x-x_{2j}(t)}}, $$
which was considered in \cite{chang2016multipeakons}.

Note that this reduction does not change the form of the Weyl function, that is, the Weyl function still possesses the form:
\begin{equation*}
W(z)=\int \frac{d\mu(x)}{x-z}, \quad d\mu=\sum_{i=1}^K
b_j \delta_{\zeta _j}, \quad 0<\zeta_1<\dots< \zeta_K, \quad 0<b_j,   \quad 1\leq j\leq K.
\end{equation*}

Let us consider the multipeakon formulae  \eqref{eq:detinversexodd}-\eqref{eq:detinversexeven} for the case of the even number of masses. It is not hard to verify that
%$$\frac{\D_k^{(l,p)}}{|\Delta_{[1,k]}(\mathbf e)|}\rightarrow \sum_{J\in\binom{[1,K]}{k}}b_J\zeta_J^{l}(\Delta_J)^2\triangleq H_k^l, \qquad \text{as} \qquad m_{2j},\ n_{2j-1}\rightarrow0, $$
%I am not getting this result; instead I am getting
$$\left(\mathbf{e}_{[1,k]}\right)^l\frac{\D_k^{(l,p)}}{|\Delta_{[1,k]}(\mathbf e)|}\rightarrow \sum_{J\in\binom{[1,K]}{l}}b_J\zeta_J^{p}\Delta_J^2\triangleq H_l^p, \qquad \text{as} \qquad m_{2j},\ n_{2j-1}\rightarrow0, $$
which leads to
$$
    \begin{aligned}
      &&x_{2K+1-2k} &\rightarrow \ln\left(\frac{1}{m_{2K+1-2k}}\cdot\frac{(H_{k}^0)^2}{H_{k}^1H_{k-1}^1}\right), && 1\leq k\leq K, \\
      &&x_{2K+2-2k} &\rightarrow\ln\left(n_{2K+2-2k}\cdot\frac{H_{k}^0H_{k-1}^0}{(H_{k-1}^1)^2}\right), && 1\leq k\leq K.
    \end{aligned}
$$
These formulae are identical to those obtained in \cite{chang2016multipeakons} for the interlacing multipeakons of 2-mCH equation. It is useful to compare
examples \ref{ex:2peakon} and \ref{ex:4peakon} in the present paper and examples 5.4, 5.5 in \cite{chang2016multipeakons} to see a concrete
manifestation of the transition from a general configuration to an interlacing one.
It is also helpful to observe that the inverse problem
for the interlacing case was solved using diagonal Pad\'{e} approximations at $\infty$
while in the present paper we are formulating our interpolation problem
at $e_j$ (see \eqref{eq:ej} all of which get moved to $\infty$ in the limit $m_{2j}\rightarrow 0,  n_{2j-1}\rightarrow 0$.
\subsection{From 2-mCH to 1-mCH}
As is known, when $v=u$, the 2-mCH \eqref{eq:m2CH} reduces to the mCH \eqref{eq:m1CH}. As for the formulae \eqref{eq:detinversexodd}-\eqref{eq:detinversexeven} in Theorem \ref{thm:inversex},
the identical mass setting
$$n_j=m_j,$$
leads to
$$v=u=\sum_{j=1}^N m_j (t)e^{-\abs{x-x_j(t)}},\qquad n=m=2\sum_{j=1}^N m_j \delta_{x_j}.$$
In this case we have
$$e_j=\frac{1}{m_{j'}^2}$$
and the formulae \eqref{eq:detinversexodd}-\eqref{eq:detinversexeven} in Theorem \ref{thm:inversex} reduce to the multipeakon formulae for the mCH equation (see  \cite[Theorem 4.21]{chang-szmigielski-m1CHlong}). Moreover, it is not hard to see that the global existence condition and the long time asymptotics in Section \ref{sec:evenpeakons} and \ref{sec:oddpeakons} also cover those for the mCH equation in \cite{chang-szmigielski-m1CHlong}.

\section*{Acknowledgments}
The first author was supported in part by the National Natural Science Foundation of China (\#11688101, 11731014, 11701550) and the Youth Innovation Promotion Association CAS. The second author was supported in part by NSERC \#163953.

\begin{appendix}
%\renewcommand{\appendixname}{Appendix~\Alph{section}}
%\renewcommand\thesection{\appendixname~\Alph{section}}
%\renewcommand\theequation{\Alph{section}.\arabic{equation}}
%\appendix
%\renewcommand{\thesection}{\Alph{chapter}.}
\section{Lax pair for the 2-mCH peakon ODEs}\label{lax_m2ch}
   The purpose of this appendix is to give a rigorous interpretation for
the distributional Lax pair of the 2-mCH equation in the generic
case, i.e. the Lax pair of the 2-mCH peakon ODEs \eqref{eq:m2CH_ode}. We note that the transition from the smooth sector to the distribution sector
   is not canonical. The argument presented below is closer to that in Appendix~A
of~\cite{chang-szmigielski-m1CHlong} than to the one used for the interlacing case of the 2-mCH \cite{chang2016multipeakons}. Let us begin by reviewing some notation needed to present the argument.

Let $\Omega_k$ denote the region $x_k(t)<x<x_{k+1}(t)$, where $x_k$
are smooth functions such that
$-\infty=x_0(t)<x_1(t)<\dots<x_{N}(t)<x_{N+1}(t)=+\infty$.

Let the function space $PC^\infty$ consist of all piecewise smooth
functions $f(x,t)$ such that the restriction of~$f$ to each
region~$\Omega_k$ is a smooth function $f_k(x,t)$ defined on ~$\Omega_k$. Then, for each fixed $t$, $f(x,t)$
defines a regular distribution~$T_f(t)$ \cite{schwartz-distributions} in the class~$\mathcal{D}'(R)$
(for simplicity we will write $f$ instead of~$T_f$). Since distributions
do not in general have values at individual points, we do not require
$f(x_k)$ to be defined.  However, the left  and right limits
are defined.  Let us denote them  $f_x(x_k-,t)$ and $f_x(x_k+,t)$  and set
\begin{equation*}
  [f](x_k)=f(x_k+,t)-f(x_k-,t),
  \qquad
  \langle f\rangle(x_k)=\frac{f(x_k+,t)+f(x_k-,t)}{2},
\end{equation*}
to denote the jump and the average, respectively.  
Denote by $f_x$ (or $f_t$) the ordinary (classical) partial derivative
with respect to~$x$ (or~$t$), and by
$\frac{\partial f_k}{\partial x}$
(or $\frac{\partial f_k}{\partial t}$) their restrictions
to~$\Omega_k$. Let $D_xf$ denote the distributional derivative with
respect to~$x$.   Then we have a well-known identity
\begin{equation*}
  D_xf=f_x+\sum_{k=1}^{N}[f](x_k)\delta_{x_k}.
\end{equation*}
Likewise, we can define $D_tf$ as a distributional limit
\begin{equation*}
  D_tf(t)=\lim_{a\rightarrow0}\frac{f(t+a)-f(t)}{a},
\end{equation*}
which can be computed explicitly to be
\begin{equation*}
  D_tf=f_t-\sum_{k=1}^{N}\dot x_{k}[f(x_k)]\delta_{x_k},
\end{equation*}
where $\dot x_{k}=dx_k/dt$.

The following formulas will be useful in further analysis:
\begin{equation} \label{eq:jumpave}
  [fg]=\langle f\rangle[g]+[f]\langle g\rangle,
  \qquad
  \langle fg\rangle=\langle f\rangle\langle g\rangle+\frac{1}{4}[f][g]
\end{equation}
\begin{equation}
  \frac{d}{dt}[f](x_k)=[f_x](x_k)\dot x_k +[f_t](x_k),
\end{equation}
\begin{equation}
  \frac{d}{dt}\langle f\rangle(x_k)=\langle f_x\rangle(x_k)\dot x_k +\langle f_t\rangle(x_k),
\end{equation}
for any $f,g\in PC^\infty$.

It is easy to see that the peakon solution $u(x,t),v(x,t)$ and the corresponding functions $\Psi_1,\ \Psi_2$ in the Lax pair
\eqref{2chlax} are piecewise smooth functions of class $PC^\infty$. Moreover, $u,v$ are continuous throughout $\R$ but $u_x, v_x,\Psi_1, \Psi_2$ have a jump at each $x_{k}$.

Let us now set $\Psi=(\Psi_1,\Psi_2)^T$, and consider an overdetermined system (see \eqref{2chlax})
\begin{equation}\label{eq:DLax-pair}
D_x\Psi=\frac{1}{2}\hat L\Psi,\qquad D_t\Psi=\frac{1}{2}\hat A\Psi,
\end{equation}
where
\begin{eqnarray}
  &&\hat L=L+2\lambda\left(\sum_{k=1}^NM_{k}\delta_{x_{k}}\right),\label{lax_peakon1}\\
  &&\hat A=A-2\lambda\left(\sum_{k=1}^NM_{k}Q(x_{k})\delta_{x_{k}}\right)\label{lax_peakon2}
\end{eqnarray}
with
\begin{eqnarray*}
  &&L=\left(
\begin{array}{cc}
     -1     & 0  \\
 0&    1 \\
\end{array}
\right),
\qquad M_k=\left(
\begin{array}{cc}
     0     & m_k  \\
 -n_k&    0 \\
\end{array}
\right),\\
&& A=\left( \begin{array}{cc}
4\lambda^{-2}+Q              &  -2\lambda^{-1}(u-u_x)\\
2\lambda^{-1}(v+v_x) & -Q
\end{array} \right)
\end{eqnarray*}
and $Q=(u-u_x)(v+v_x)$.
We observe that the above splitting of the Lax pair corresponds to the distributional splitting into a regular distribution and a distribution 
with a singular support.
In particular the singular part of (\eqref{lax_peakon1})  requires that we multiply $M_k\Psi=(m_k\Psi_2,-n_k\Psi_1)$ by $\delta_{x_{k}}$. This is not defined, and to define this multiplication we need to assign some values to $\Psi_1,\Psi_2$ at $x_k$. This, in principle, is an arbitrary procedure but
we want it to reflect the local behaviour of $\Psi$ around $x_k$ and postulate that our choice of values of $\Psi$ depends linearly on
the right and left hand limits.

Likewise, for the $t$-Lax equation (\ref{lax_peakon2}) to be defined as a distribution equation, $QM_k\Psi=(Qm_k\Psi_2,-Qn_k\Psi_1)$ needs to be a multiplier of $\delta_{x_{k}}$. Thus, in the same sense as above, the values of $Q(x_{k})$ need to be assigned as well.

Henceforth, we will refer to these assignments as \textit{regularizations},
even though this name in the theory of distributions refers usually to
re-defining divergent integrals.

We now specify more concretely a family of regularizations we consider.  We will refer to them as \textit{invariant regularizations}.  The reason for that
name is explained in \cite[Appendix A] {chang-szmigielski-m1CHlong}.

\begin{definition} \label{def:invariantreg}
An invariant regularization of
the Lax pair \eqref{eq:DLax-pair} consists of specifying
the values of $\alpha, \beta \in \R$ and $Q(x_k)=(u-u_x)(v+v_x)(x_k)$ in the formulas
\begin{equation*}
\begin{aligned}
\Psi(x)\delta _{x_k}&=\Psi(x_k) \delta_{x_k}, \\
\Psi(x_k)&=\alpha \jump{\Psi}(x_k)+\beta\avg{\Psi}(x_k),  \\
Q(x)\delta_{x_k}&=Q(x_k) \delta_{x_k}.
\end{aligned}
\end{equation*}
\end{definition}

\begin{theorem}\label{thm:invreg}
Let $m$ and $n$ be the discrete measures associated to $u$ and $v$ defined by \eqref{eq:peakonansatz}.
Given an invariant regularization in the sense of Definition \ref{def:invariantreg} the
distributional Lax pair \eqref{eq:DLax-pair} is compatible, i.e. $D_tD_x\Psi=D_xD_t\Psi$,  if and only if
the following conditions hold:

\begin{equation}
\beta =1, \qquad  \alpha=\pm \tfrac12.   \label{eq:betaalpha}
\end{equation}
Then
\begin{subequations}
\begin{align}
\dot m_k&=-m_k(Q(x_k)-\avg{Q}(x_k)), \label{eq:dotmk}\\
\dot n_k&=n_k(Q(x_k)-\avg{Q}(x_k)), \label{eq:dotnk}\\
\dot x_k&=Q(x_k). \label{eq:dotxk}
\end{align}
\end{subequations}

\end{theorem}
\begin{proof}
The proof follows almost \textit{verbatim} \cite[Theorem A.2]{chang-szmigielski-m1CHlong} (also see \cite{chang2016multipeakons,hls}). We omit the details.
\end{proof}
Now we can spell out the connection of the regularization problem to the content of the present paper.
\begin{corollary}\label{cor:disrlax}
Let $m$ and $n$ be the discrete measures associated to $u$ and $v$ defined by \eqref{eq:peakonansatz}. Suppose that $Q(x_k),\Psi(x_k)$ are
assigned values
$$
Q(x_k)=\avg{Q}(x_k),\quad \Psi(x_k)=\Psi(x_k+),
$$
or
$$Q(x_k)=\avg{Q}(x_k), \quad \Psi(x_k)=\Psi(x_k-).  $$
For either case, the compatibility condition $D_tD_x\Psi = D_xD_t\Psi$
for the distributional Lax pair \eqref{eq:DLax-pair} reads
$$
\dot m_k=0,\quad \dot n_k=0,\quad \dot x_k=\avg{Q}(x_k).
$$
\end{corollary}

\section{Proof of Theorem \ref{thm:global} }\label{app:globaleven}
\begin{proof}
It suffices to ensure that the ordering conditions $x_1<x_2<\cdots<x_{2K}$ hold for all time $t$.    We write these conditions as:
\begin{subequations}
\begin{align}
&&x_{(k+1)'}&<x_{k'}, & \text { for all odd }k, \qquad  1\leq k\leq 2K-1, \label{eq:order1}\\
&&x_{(k+2)'}&<x_{(k+1)'}, & \text{ for all odd }k, \qquad  1\leq k\leq 2K-3, \label{eq:order2}
\end{align}
\end{subequations}
and use equations \eqref{eq:detinversexodd}, \eqref{eq:detinversexeven} to
obtain equivalent conditions
\begin{subequations}
\begin{align}
\frac{1}{m_{(k+1)'}n_{k'}}<\frac{\D_{k+1}^{(\frac{k+1}{2},1)} \D_{k-1}^{(\frac{k-1}{2},0)}}{\D_{k+1}^{(\frac{k+1}{2},0)} \D_{k-1}^{(\frac{k-1}{2},1)}},&& \text { for all odd }k, \  1\leq k\leq 2K-1, \label{eq:order1bis} \\
\frac{1}{m_{(k+2)'}n_{(k+1)'}}<\frac{\D_{k+2}^{(\frac{k+1}{2},1)} \D_{k}^{(\frac{k+1}{2},0)}}{\D_{k+2}^{(\frac{k+3}{2},0)} \D_{k}^{(\frac{k-1}{2},1)}},&& \text{ for all odd } k, \ 1\leq k\leq 2K-3.  \label{eq:order2bis}
\end{align}
\end{subequations}

However, \eqref{eq:Dklp} implies that the inequality
\begin{equation}\label{eq:order1step3}
\frac{\D_{k+1}^{(\frac{k+1}{2},1)} \D_{k-1}^{(\frac{k-1}{2},0)}}{\D_{k+1}^{(\frac{k+1}{2},0)} \D_{k-1}^{(\frac{k-1}{2},1)}}>\frac{\zeta_1^{\frac{k+1}{2}}}{\zeta_K^{\frac{k-1}{2}}}
\end{equation}
holds uniformly in $t$ and
if we impose
\begin{equation*}
\frac{1}{m_{(k+1)'}n_{k'}}<\frac{\zeta_1^{\frac{k+1}{2}}}{\zeta_K^{\frac{k-1}{2}}}\text { for all odd }k, \qquad   k\leq 2K-1,
\end{equation*}
then \eqref{eq:order1} holds.

Now we focus on the second inequality, namely \eqref{eq:order2bis}, which is valid whenever $K\geq 2$.
It is convenient to consider a slightly more general expression, namely,
\begin{equation*}
\frac{\D_{k+2}^{(l,1)} \D_{k}^{(l,0)}}{\D_{k+2}^{(l+1,0)} \D_{k}^{(l-1,1)}}, \qquad 1\leq l \leq K-1,
\end{equation*}
for which, using similar steps to those to in \cite{chang-szmigielski-m1CHlong}, we can prove the bound
\begin{equation} \label{eq:order2step4}
\begin{split}
\frac{\D_{k+2}^{(l,1)} \D_{k}^{(l,0)}}{\D_{k+2}^{(l+1,0)} \D_{k}^{(l-1,1)}}> \frac{\zeta_1^{l}}{\zeta_K^{l-1}}
\frac{\textnormal{min}_j(\zeta_{j+1}-\zeta_j)^{2(l-1)}}{l(\zeta_K-\zeta_1)^{2l}}\cdot(e_{k+1}+\zeta_1)(e_{k+2}+\zeta_1)\\=\frac{\zeta_1^{l}}{\zeta_K^{l-1}}
\frac{\textnormal{min}_j(\zeta_{j+1}-\zeta_j)^{2(l-1)}}{l(\zeta_K-\zeta_1)^{2l}}\cdot\frac{(1+m_{(k+1)'}n_{(k+1)'}\zeta_1)(1+m_{(k+2)'}n_{(k+2)'}\zeta_1)}{m_{(k+1)'}n_{(k+1)'}m_{(k+2)'}n_{(k+2)'}}.
\end{split}
\end{equation}
Hence, specializing to $l=\frac{k+1}{2}$ and assuming
\begin{align*}
&\frac{1}{m_{(k+2)'}n_{(k+1)'}} < \\
&\frac{\zeta_1^{\frac{k+1}{2}}}{\zeta_K^{\frac{k-1}{2}}}
\frac{2\, \textnormal{min}_j(\zeta_{j+1}-\zeta_j)^{k-1}}{(k+1)(\zeta_K-\zeta_1)^{k+1}}\cdot\frac{(1+m_{(k+1)'}n_{(k+1)'}\zeta_1)(1+m_{(k+2)'}n_{(k+2)'}\zeta_1)}{m_{(k+1)'}n_{(k+1)'}m_{(k+2)'}n_{(k+2)'}},
\end{align*}
we obtain that \eqref{eq:order2bis} and thus \eqref{eq:order2} hold.

Finally, after rewriting
the last condition as:
\begin{align*}
&\frac{n_{(k+2)'}m_{(k+1)'}}{(1+m_{(k+1)'}n_{(k+1)'}\zeta_1)(1+m_{(k+2)'}n_{(k+2)'}\zeta_1)}<
\frac{\zeta_1^{\frac{k+1}{2}}}{\zeta_K^{\frac{k-1}{2}}}
\frac{2\, \textnormal{min}_j(\zeta_{j+1}-\zeta_j)^{k-1}}{(k+1)(\zeta_K-\zeta_1)^{k+1}}, \\
&\hspace{8cm} \text{ for all odd } k, \  k\leq 2K-3
\end{align*}
we obtain the second condition \eqref{eq:seccond}.
\end{proof}

\section{Proof of Theorem \ref{thm:global(odd)} }\label{app:globalodd}
\begin{proof}
Again, it suffices to ensure that the ordering conditions $x_1<x_2<\cdots<x_{2K+1}$ hold for all time $t$.
  We write these conditions in an equivalent form:
\begin{subequations}
\begin{align}
x_{(k+1)'}<x_{k'}, && \text { for all odd }k, \qquad  1\leq k\leq 2K-1, \label{eq:order1-odd}\\
x_{(k+2)'}<x_{(k+1)'}, && \text{ for all odd }k, \qquad  1\leq k\leq 2K-1, \label{eq:order2-odd}
\end{align}
\end{subequations}
and use equations \eqref{eq:detinversexodd}, \eqref{eq:detinversexeven} to
obtain
\begin{subequations}
\begin{align}
&\frac{1}{m_{(k+1)'}n_{k'}}<\frac{\D_{k+1}^{(\frac{k+1}{2},1)}(c) \D_{k-1}^{(\frac{k-1}{2},0)}(c)}{\D_{k+1}^{(\frac{k+1}{2},0)} (c)\D_{k-1}^{(\frac{k-1}{2},1)}(c)}, \text { for all odd }k,  1\leq k\leq 2K-1, \label{eq:order1bis-odd} \\
&\frac{1}{m_{(k+2)'}n_{(k+1)'}}<\frac{\D_{k+2}^{(\frac{k+1}{2},1)}(c) \D_{k}^{(\frac{k+1}{2},0)}(c)}{\D_{k+2}^{(\frac{k+3}{2},0)}(c) \D_{k}^{(\frac{k-1}{2},1)}(c)}, \text{ for all odd }k,  1\leq k\leq 2K-1.   \label{eq:order2bis-odd}
\end{align}
\end{subequations}

From \eqref{eq:Dklpodd1-c-dep} and \eqref{eq:Dklpodd2-c-dep}
we obtain
\begin{equation*}
\begin{split}
\frac{\D_{k+1}^{(\frac{k+1}{2},1)}(c) \D_{k-1}^{(\frac{k-1}{2},0)}(c)}{\D_{k+1}^{(\frac{k+1}{2},0)}(c) \D_{k-1}^{(\frac{k-1}{2},1)}(c)}=\frac{ \big(\D_{k+1}^{(\frac{k+1}{2}, 1)}(0)+c\D_{k+1}^{(\frac{k-1}{2}, 1)}(0)\big) \D_{k-1}^{(\frac{k-1}{2},0)}(0)}{
\D_{k+1}^{(\frac{k+1}{2},0)}(0)\big(\D_{k-1}^{(\frac{k-1}{2},1)}(0)+c\D_{k-1}^{(\frac{k-3}{2},1)}(0)\big)}\\=\frac{ \D_{k+1}^{(\frac{k+1}{2}, 1)}(0)\D_{k-1}^{(\frac{k-1}{2},0)}(0)+c\D_{k+1}^{(\frac{k-1}{2}, 1)}(0)\D_{k-1}^{(\frac{k-1}{2},0)}(0)}{
\D_{k+1}^{(\frac{k+1}{2},0)}(0)\D_{k-1}^{(\frac{k-1}{2},1)}(0)+c\D_{k+1}^{(\frac{k+1}{2},0)}(0)\D_{k-1}^{(\frac{k-3}{2},1)}(0)}\stackrel{def}{=} \frac {\mathcal{A}_1+\mathcal{B}_1}{\mathcal{A}_2+\mathcal{B}_2},
\end{split}
\end{equation*}
where $\mathcal{B}_2=0$ for $k=1$.  Hence the ratios $\frac{\mathcal{A}_1}{\mathcal{A}_2}, \frac{\mathcal{B}_1}{\mathcal{B}_2}$ satisfy (uniform in $t$) bounds
\begin{align*}
&\frac{\mathcal{A}_1}{\mathcal{A}_2}>\frac{\zeta_1^{\frac{k+1}{2}}}{\zeta_K^{\frac{k-1}{2}}}\stackrel{def}{=}\alpha, \\
&\frac{\mathcal{B}_1}{\mathcal{B}_2}>\frac{\zeta_1^{\frac{k-1}{2}}}{\zeta_K^{\frac{k-3}{2}}}\frac{2\, \textrm{min}_j(\zeta_{j+1}-\zeta_j)^{k-3}}{(k-1)(\zeta_K-\zeta_1)^{k-1}}\frac{(1+m_{k'}n_{k'}\zeta_1)(1+m_{(k+1)'}n_{(k+1)'}\zeta_1)}{m_{k'}n_{k'} m_{(k+1)'}n_{(k+1)'}}\\
&\hspace{0.5cm}\stackrel{def}{=}\beta,
\end{align*}
by equations \eqref{eq:order1step3} and \eqref{eq:order2step4}, respectively,
with the convention that $\beta=\infty$ for the special case $k=1$.
Thus
\begin{equation*}
\textrm{min}\{\alpha, \beta\}<\frac{\D_{k+1}^{(\frac{k+1}{2},1)}(c) \D_{k-1}^{(\frac{k-1}{2},0)}(c)}{\D_{k+1}^{(\frac{k+1}{2},0)}(c) \D_{k-1}^{(\frac{k-1}{2},1)}(c)}
\end{equation*}
holds uniformly in $t$ and
if we impose
\begin{equation*}
\frac{1}{m_{(k+1)'}n_{k'}}<\textrm{min}\{\alpha, \beta\}\qquad \text { for all odd }k, \qquad   1\leq k\leq 2K-1,
\end{equation*}
then equations \eqref{eq:order1-odd} will hold automatically.

Now we turn to the second inequality \eqref{eq:order2bis-odd}. From
Corollary \ref{cor:Dklp-c-dependence} we obtain
\begin{equation*}
\frac{\D_{k+2}^{(\frac{k+1}{2},1)}(c) \D_{k}^{(\frac{k+1}{2},0)}(c)}{\D_{k+2}^{(\frac{k+3}{2},0)}(c) \D_{k}^{(\frac{k-1}{2},1)}(c)}>\frac {\mathcal{A}_1+\mathcal{B}_1}{\mathcal{A}_2+\mathcal{B}_2}.
\end{equation*}
Since
\begin{align*}
&\frac{\mathcal{A}_1}{\mathcal{A}_2}>\frac{\zeta_1^{\frac{k+1}{2}}}{\zeta_K^{\frac{k-1}{2}}}=\alpha, \\
&\frac{\mathcal{B}_1}{\mathcal{B}_2}>\frac{\zeta_1^{\frac{k+1}{2}}}{\zeta_K^{\frac{k-1}{2}}}\frac{2\, \textrm{min}_j(\zeta_{j+1}-\zeta_j)^{k-1}}{(k+1)(\zeta_K-\zeta_1)^{k+1}}\frac{(1+m_{(k+1)'}n_{(k+1)'}\zeta_1)(1+m_{(k+2)'}n_{(k+2)'}\zeta_1)}{m_{(k+1)'}n_{(k+1)'} m_{(k+2)'}n_{(k+2)'}}\\
&\hspace{0.5cm}\stackrel{def}{=}\beta_1,
\end{align*}
then
\begin{equation*}
\textrm{min}\{\alpha, \beta_1\}<\frac{\D_{k+2}^{(\frac{k+1}{2},1)}(c) \D_{k}^{(\frac{k+1}{2},0)}(c)}{\D_{k+2}^{(\frac{k+3}{2},0)}(c) \D_{k}^{(\frac{k-1}{2},1)}(c)}
\end{equation*}
is satisfied.
Thus inequality
\begin{equation*}
\frac{1}{m_{(k+2)'}n_{(k+1)'}}<\textrm{min}\{\alpha, \beta_1\}, \qquad  \text{ for all odd } k, \qquad 1\leq k\leq 2K-1,
\end{equation*}
implies \eqref{eq:order2bis-odd} and, consequently, \eqref{eq:order2-odd},
thereby proving the claim.

\end{proof}

\end{appendix}

\strut\hfill

\end{document}